\documentclass[11pt]{article}

\usepackage{amsmath}
\usepackage{mathtools}
\usepackage{fullpage}
\usepackage{amsfonts}
\usepackage{amssymb}
\usepackage{graphicx}
\usepackage{color}
\usepackage{theoremref}
\usepackage{overpic}
\usepackage{times}
\newtheorem{theorem}{Theorem}[section]

\newtheorem{lemma}[theorem]{Lemma}

\newtheorem{remark}[theorem]{Remark}
\newtheorem{definition}[theorem]{Definition}

\numberwithin{equation}{section}
\newenvironment{proof}[1][Proof]{\noindent\textbf{#1.} }{\ \rule{0.5em}{0.5em}}

\renewcommand{\epsilon}{\varepsilon}
\newcommand{\ket}[1]{\mathop{\left|#1\right>}\nolimits}

\newcommand{\kb}[2]{\left| #1\right\rangle\!\left\langle #2 \right|}
\newcommand{\Tra}[1]{\mathop{{\mathrm{Tr}}_{#1}}}
\newcommand{\Tr}[2]{\mathop{{\mathrm{Tr}}_{#1}} \left(#2\right) }

\def\N{\mathcal{N}}

\def\R{\mathcal{R}}

\def\D{\mathcal{D}}

\def\sA{\mathsf{A}}

\def\sB{\mathsf{B}}

\def\QIC{\mathrm{QIC}}

\newcommand{\suppress}[1]{}
\newcommand{\set}[1]{\left\{ #1 \right\}}
\newcommand{\size}[1]{\left| #1 \right|}

\newcommand{\xor}{\oplus}

\newcommand{\dyck}{\textsc{Dyck}}
\newcommand{\rA}{{\mathrm A}}
\newcommand{\rB}{{\mathrm B}}
\newcommand{\rF}{{\mathrm F}}
\newcommand{\rH}{{\mathrm H}}
\newcommand{\rin}{{\mathrm{in}}}
\newcommand{\rout}{{\mathrm{out}}}
\newcommand{\QCC}{{\mathrm{QCC}}}
\DeclareMathOperator{\rI}{{\mathrm I}}
\newcommand{\bI}{{\mathbb I}}
\DeclareMathOperator{\fh}{{\mathfrak h}}
\newcommand{\cA}{{\mathcal A}}
\newcommand{\Asc}{{\textsc{Ascension}}}
\newcommand{\asc}{{\textsc{Asc}}}
\newcommand{\ai}{{\textsc{AI}}}

\newcommand{\rO}{{\mathrm O}}

\newcommand{\eqdef}{\coloneqq}

\makeatletter
\let\@copyrightspace\relax
\makeatother

\begin{document}

\title{Augmented Index and \\ Quantum Streaming Algorithms 
for \dyck(2)~\thanks{An abridged version of this article appeared as
Ref.~\cite{NT17}.}
}

\author{
Ashwin Nayak \thanks{Department of Combinatorics and Optimization,
and Institute for Quantum Computing,
University of Waterloo, 200 University Ave.\ W., Waterloo, ON,
N2L~3G1, Canada. Email: \texttt{ashwin.nayak@uwaterloo.ca}.
Research supported in part by NSERC Canada.}
\and 
Dave Touchette \thanks{Department of Computer Science, and Institut
Quantique, Universit{\'e} de Sherbrooke, 2500 Boulevard de l'Universit{\'e},
Sherbrooke, QC,  J1K~2R1, Canada. Email: \texttt{touchette.dave@gmail.com}.
This work was done while the author was at Institute for Quantum 
Computing and Department of Combinatorics and Optimization, University 
of Waterloo, and the Perimeter Institute of Theoretical Physics. Research
supported in part by NSERC, partly via the PDF program, CIFAR, and by
Industry Canada. IQC and PI are supported in part by the Government of
Canada and the Province of Ontario.
}
}


\date{August 16, 2022}

\maketitle

\begin{abstract}
We show how two recently developed quantum information theoretic tools can be 
applied to obtain lower bounds on quantum information complexity.
We also develop new tools with potential for broader applicability,
and use them to establish a lower bound on the quantum information
complexity for the Augmented Index function on an easy distribution.
This approach allows us to handle superpositions rather than
distributions
over inputs, the main technical challenge faced  previously.
By providing a quantum generalization of the argument of Jain and
Nayak~[IEEE TIT'14], we leverage this to obtain a lower bound on the
space complexity of multi-pass, unidirectional
quantum streaming algorithms for the \dyck(2) language.
\end{abstract}

\thispagestyle{empty}

\newpage

\setcounter{page}{1}

\section{Introduction}

\subsection{Streaming Algorithms and Augmented Index}

The first \emph{bona fide\/} quantum computers that are built are likely 
to involve a few hundred qubits, and
be limited to short computations. This prompted much research into
the capabilities of space bounded quantum computation, especially of
quantum finite automata, during the early development of the theory 
of quantum computation (see, e.g., Refs.~\cite{MC00,KW97,AF98,ANTV02}).
More recently, this has led to the investigation
of \emph{quantum streaming
algorithms\/}~\cite{LeGall09,GavinskyKKRW08,Blume-KohoutCG14,Montanaro16}.

Streaming algorithms were originally proposed as a means of processing massive
real-world data that cannot be stored in their entirety in
computer memory~\cite{Muthukrishnan05}. Instead of random access to the input data, these
algorithms receive the input in the form of a \emph{stream\/}, i.e., one 
input symbol at a time. The algorithms attempt to solve some information 
processing task using as little space and time as possible, on
occasion using more than one sequential pass over the input stream.

Streaming algorithms provide a natural framework for studying simple,
small-space quantum computation beyond the scope of quantum finite
automata. Some of the works mentioned above (e.g.,
LeGall~\cite{LeGall09}) show how quantum streaming algorithms can accomplish
certain specially crafted tasks with exponentially smaller space,
as compared with classical algorithms. This led Jain and
Nayak~\cite{JainN14} to ask how much more efficient such quantum
algorithms could be for other, more natural problems. They focused on
\dyck(2), a well-studied and important problem from formal language
theory~\cite{ChomskyS63}.  \dyck(2) consists of all well-formed expressions
with two types of parenthesis, denoted below by~$a, \overline{a}$
and~$b,\overline{b}$, with the bar indicating a closing parenthesis.
More formally,
$\dyck(2)$ is the language over the alphabet~$\Sigma = \set{a,
\overline{a},
b,\overline{b}}$ defined recursively as
\[
\dyck(2) \quad = \quad \epsilon + \bigl( a \cdot \dyck(2) \cdot \overline{a}
                  + b \cdot \dyck(2) \cdot \overline{b} \bigr)
                  \cdot \dyck(2) \enspace,
\]
where~$\epsilon$ is the empty string, `$\cdot$' indicates concatenation
of strings (or subsets thereof) and `$+$' denotes set union.

The related problem of recognizing whether a given expression with
parentheses is well-formed was originally studied by Magniez, Mathieu,
and Nayak~\cite{MagniezMN14} in the context of \emph{classical\/} streaming
algorithms. They discovered a remarkable phenomenon, that providing
\emph{bi-directional\/} access to the input stream leads to an
exponentially more space-efficient algorithm. They presented
a randomized streaming algorithm that makes one pass over the input,
uses~$\rO(\sqrt{n\log n}\,)$
bits, and makes polynomially small probability of error
to determine membership of expressions of length~$\rO(n)$ in \dyck(2).
Moreover, they proved that this space bound is
optimal for error at most~$1/(n \log n)$, and conjectured that a similar
polynomial space bound holds for multi-pass algorithms as well.
Magniez \emph{et al.}\ complemented this with a second randomized
algorithm that makes two passes in \emph{opposite\/} directions over 
the input, uses only~$\rO(\log^2 n)$ space, and has polynomially 
small probability of error. Later, two sets of
authors~\cite{ChakrabartiCKM10,JainN14} independently and
concurrently proved the conjectured hardness of \dyck(2) for classical
multi-pass streaming algorithms that read the input only in one
direction.  They showed that any
unidirectional randomized~$T$-pass streaming algorithm that
recognizes length~$n$ instances of \dyck(2) with a constant
probability of error uses space~$\Omega(\sqrt{n}/T)$.

The space lower bounds for \dyck(2) established in
Refs.~\cite{MagniezMN14,ChakrabartiCKM10,JainN14} all rely on a
connection with a two-party communication problem, Augmented Index,
a variant of Index, a basic problem in two-party communication
complexity. In
the Index function problem, one party, Alice, is given a
string~$x \in \set{0,1}^n$, and the other party, Bob, is given an
index~$k \in [n]$, for some positive integer~$n$. Their goal is to
communicate with each other and compute~$x_k$, the~$k$th bit of the
string~$x$. In the Augmented Index function problem, Bob is given
the prefix~$x[1,k-1]$ (the first~$k-1$ bits of~$x$) and a bit~$b$
in addition to the index~$k$. The
goal of the two parties is to determine if~$x_k = b$ or not.
The three works cited above (see also~\cite{ChakrabartiK11}) 
all prove information cost trade-offs for Augmented Index. As a result, 
in any bounded-error protocol for the function, either Alice
reveals~$\Omega(n)$ information about her input~$x$,
or Bob reveals~$\Omega(1)$ information about the index~$k$ 
(even under an easy distribution, the uniform
distribution over zeros of the function).

Motivated by the abovementioned works, Jain and Nayak~\cite{JainN14} 
studied quantum protocols for Augmented
Index. They defined a notion of quantum information cost
for distributions with a limited form of
dependence, and then arrived at a similar tradeoff as in the classical
case. This result, however, does not imply a lower bound on the space
required by quantum streaming algorithms for \dyck(2). The issue is that
the reduction from low information cost protocols for Augmented Index to
small space streaming algorithms breaks down in the quantum case (for
the notion of quantum information cost they proposed). This left open
the possibility of more efficient unidirectional quantum streaming
algorithms.

\subsection{Overview of Results}

We establish the following lower bound on the space complexity of $T$-pass, 
unidirectional quantum streaming algorithms for \dyck(2),
thus solving the question posed by Jain and Nayak~\cite{JainN14}.
\begin{theorem}
\label{th:infmainlbstream}
For any $T \geq 1$, any unidirectional $T$-pass quantum streaming algorithm 
that recognizes $\dyck(2)$ with a constant probability of error uses
space $\Omega (\sqrt{n}/ T^3)$ on length~$n$ instances of the problem.
\end{theorem}

The space bound above holds for a general model for quantum streaming
algorithms, one in which the computation is characterized by
arbitrary quantum operations. In particular, the computation may be
non-unitary, and may use ``on-demand'' ancillary qubits in addition to
the allowed work space. Some earlier work showing strong limitations
of bounded space, such as that on quantum finite
automata~\cite{ANTV02}, assumed unitary evolution.

Theorem~\ref{th:infmainlbstream} shows that, possibly up to logarithmic
factors and the dependence on the number of passes, 
quantum streaming algorithms are no more efficient than classical ones for this problem. 
In particular, this provides the first natural example for which classical bi-directional 
streaming algorithms perform exponentially better than unidirectional quantum streaming algorithms.

Theorem~\ref{th:infmainlbstream} is a consequence of a lower bound that
we establish on a measure of quantum information cost introduced
by Touchette~\cite{Touchette15}. (Henceforth, we use the term
``quantum information cost'' without any qualification to refer to
this notion.) We consider this cost for any quantum protocol $\Pi$
computing the Augmented Index function, with respect to an ``easy''
distribution~$\mu_0$: the
uniform distribution over the zeros of the function. Due to the asymmetry 
of the Augmented Index function, we distinguish between the amount of 
information Alice transmits to Bob, denoted $\QIC_{\rA \rightarrow \rB}
(\Pi, \mu_0)$ 
and the amount of information Bob transmits to Alice, denoted
$\QIC_{\rB \rightarrow \rA} (\Pi, \mu_0)$; see
Section~\ref{sec-qic} for formal definitions for these notions.
Our main technical contributions are in proving the following trade-off.
\begin{theorem}
\label{th:infmainlbaugind}
In any $t$-round quantum protocol $\Pi$ computing the Augmented Index function $f_n$ with 
constant error $\epsilon \in [0, 1/4)$ on any input, either $\QIC_{\rA \rightarrow \rB} (\Pi, \mu_0) \in \Omega (n / t^2)$ 
or $\QIC_{\rB \rightarrow \rA} (\Pi, \mu_0) \in \Omega (1 /t^2)$.
\end{theorem}
A more precise statement is presented as Theorem~\ref{th:formainlbaugind}.
As in previous works, establishing a lower bound on the quantum
information cost for such an easy distribution is necessary; the
direct sum argument that links quantum streaming algorithms to
quantum protocols for Augmented Index crucially hinges on this.
(This phenomenon is common in such direct sum arguments.)

The high level intuition underlying the proof of
Theorem~\ref{th:infmainlbaugind}
and its structure is the same as that in Ref.~\cite{JainN14}.
There are two principal challenges in their approach, and the choice of
an appropriate measure of information cost is fundamental to
overcoming both challenges.
The first challenge is a direct sum argument that relates streaming
algorithms for \dyck(2) and communication protocols for Augmented Index.
The second challenge is establishing an information cost trade-off for
Augmented Index. Jain and Nayak considered several notions of
information cost, each one of which was effective in addressing one
challenge \emph{but not the other\/}.
This was further complicated by the intrinsic correlation
of the inputs for Augmented Index held by the two parties. Indeed, an
important motivation behind the notion of quantum information cost used
in Ref.~\cite{JainN14} is the desire to avoid leaking information about
the inputs by virtue of their preparation in superposition, instead of
exchanging information through interaction alone.
The notion they analyzed in detail admits an information cost trade-off,
but not a connection between streaming algorithms and low information
protocols. In particular, the notion does not seem to be bounded by
communication complexity.

Quantum information cost, as proposed by Touchette~\cite{Touchette15},
turns
out to be a suitable choice for quantifying the information content of
messages in our context. It is defined in terms of
conditional mutual information, conditioned on the
recipient's quantum state. Thus, this notion naturally avoids the
difficulties
arising from the intrinsic correlation between the two parties' inputs.
It is also relatively simple to derive low quantum
information cost protocols for Augmented Index from small-space streaming
algorithms for \dyck(2), through a direct sum argument.
Remarkably, the properties of quantum information cost allow us to
execute the reduction even for algorithms whose computation involves
arbitrary quantum operations, including non-unitary evolution.
However, a quantum information cost trade-off for Augmented Index still
presents significant obstacles. In order to overcome these,
we develop new tools for quantum communication complexity that we believe
have broader applicability.

One tool is a generalization of the well-known Average Encoding Theorem
of (classical and) quantum complexity theory~\cite{KNTZ07},
which formalizes the intuition that weakly correlated systems are
nearly independent. We call this generalized version the
\emph{Superposition-Average Encoding Theorem\/}, as it allows us to
handle arbitrary superpositions over inputs to quantum communication
protocols (as opposed to classical distributions over inputs).
The proof of this theorem builds on the breakthrough result by Fawzi
and Renner~\cite{FawziR15}, linking conditional quantum mutual information
to the optimal recovery map acting on the conditioning system.
Note that there is an obvious generalization of the Average Encoding
Theorem to an analogous result for conditional quantum mutual information
implied by the Fawzi-Renner inequality together with the Uhlmann theorem.
This \emph{cannot\/} directly be used in a proof \emph{{\`a} la\/}
Ref.~\cite{JainN14}. For one, such a generalization would give
us a unitary operation that acts on one part of a (pure)
``reconstructed''
state, and maps it to a state close to a target state. The hybrid
argument in Ref.~\cite{JainN14} relies on the commutativity of such
unitary operations corresponding to  successive messages in a
protocol, whereas the operations obtained by the obvious generalization do not commute.

Another key ingredient in the proof of Theorem~\ref{th:infmainlbaugind} 
is a \emph{Quantum Cut-and-Paste Lemma}, a variant of a technique
used in Refs.~\cite{JainRS03b, JainN14}, that
allows us to deal with easy distributions over inputs. The cut-and-paste lemma 
for randomized communication protocols connects the distance between transcripts obtained by 
running protocols on inputs chosen from a two-by-two rectangle $\{ x, x^\prime \} \times \{y, y^\prime \}$.
The cut-and-paste lemma is very powerful, and a direct quantum analogue does not hold. We can 
nevertheless obtain the following weaker variant, linking any four possible pairs of inputs in a two-by-two rectangle: 
if the states for a fixed input~$y$ to Bob are close up to a local unitary operator on Alice's side and the states for a fixed input~$x$ to Alice
are close up to a local unitary operator on Bob's side, then, up to local unitary
operators on Alice's and Bob's sides, the states for 
all pairs $(x^{\prime \prime}, y^{\prime \prime})$ of inputs in the rectangle $\{x, x^\prime \} \times \{y, y^\prime \}$ 
are close to each other.
This lemma allows us to link output states of protocols on inputs from an easy distribution, all mapping to the same output value, to an output state corresponding to a different output value. This
helps derive a contradiction to the assumption of low quantum
information cost, as states
corresponding to different outputs are distinguishable with constant
probability.

We go a step further with the quantum information cost trade-off.
We provide an alternative way to achieve a similar result, by using a
notion of information cost tailored to the Augmented Index problem.
An important stepping stone in this approach is the recently developed
\emph{Information Flow Lemma}
due to Lauri\`ere and Touchette~\cite{LT17-information-flow, LT17-information-flow-arxiv}.
The lemma allows us to track the transfer of information due to
interaction in quantum protocols, and provides insight into how
information might be leaked due to a superposition over inputs. By
conditioning on a suitable classical register, we avoid such leakage
of information. Pushing these ideas further, we are able to
bring the Average Encoding Theorem to bear in this context as well.
This helps us obtain a slightly better round-dependence in the
information cost trade-off.

\paragraph{Organization.}
Background and definitions related to quantum communication,
information theory, and streaming algorithms are presented in
Section~\ref{sec:prelim}. We then adapt, in Section~\ref{sec-reductions},
the known two-step reduction from Augmented Index to \dyck(2) to
the new notion for information cost due to Touchette~\cite{Touchette15}
and to the general model for streaming algorithms that we define.
The technical tools that we develop and use are presented in
Section~\ref{sec:tools}.
The main lower bound on the quantum information cost of the Augmented
Index function is derived in Section~\ref{sec:QIC-lb-AI}.
A lower bound with a slightly better dependence on the number of
rounds is presented in Section~\ref{sec:altlbqic}.

\paragraph{Acknowlegements.}
We are grateful to Mark Braverman, Ankit Garg, Young Kun Ko, and Jieming Mao for useful discussions related to the development of the Superposition-Average Encoding Theorem and Quantum Cut-and-Paste Lemma.

\section{Preliminaries}
\label{sec:prelim}

We refer the reader to text books such as~\cite{Watrous18,Wilde13} 
for standard concepts and the associated notation from quantum information.

\subsection{Quantum Communication Complexity}
\label{sec-qcc}

We use the following notation for interactive communication between two
parties, called Alice and Bob by convention. For a register~$S$, we denote the set of its possible quantum states by~$\D(S)$. An $M$-message protocol $\Pi$ 
for a task with input registers $A_{\mathrm{in}}B_{\mathrm{in}}$
and output registers $A_{\mathrm{out}}B_{\mathrm{out}}$ is defined by a sequence of isometries $U_1, \dotsc, U_{M + 1}$ 
along with a pure state $\psi \in \D (T^\rA T^\rB)$ shared between Alice and
Bob, for some arbitrary but finite 
dimensional registers $T^\rA T^\rB$. We refer to~$\psi$ as the
shared entanglement. We have $M+1$ isometries in 
an~$M$-message protocol, as one isometry is applied before each message,
and a final isometry is applied after the last message is 
received. We assume that Alice sends the first message. 
In the case of even $M$, the
registers $A_1, A_3, \dotsc, A_{M - 1}, A^\prime$ are held by Alice, the registers~$B_2, B_4, \dotsc, B_{M - 2}, B^\prime$ 
are held by Bob, and the registers~$C_1C_2C_3\cdots C_M$ represent the
quantum messages exchanged by Alice and Bob. The~$M+1$ isometries act on
these registers as indicated by the superscripts below (also see Figure~\ref{fig:int_mod}):
\begin{align*}
U_1^{A_{\mathrm{in}} T^\rA \rightarrow A_1 C_1},\qquad &U_2^{B_{\mathrm{in}}
T^\rB C_1 \rightarrow B_2 C_2},\qquad U_3^{A_1 C_2 \rightarrow A_3 C_3},\qquad U_4^{B_2 C_3 \rightarrow B_4 C_4},\notag\\
\;\cdots,\qquad & U_{M}^{B_{M - 2} C_{M - 1} \rightarrow B_{\mathrm{out}}
B^\prime C_{M}},\qquad U_{M + 1}^{A_{M - 1} C_M \rightarrow A_{\mathrm{out}}
A^\prime} \enspace.
\end{align*}
We adopt the convention that, at the outset, 
$A_0 = A_{\mathrm{in}} T^\rA$, $B_0 = B_{\mathrm{in}} T^\rB$; for odd $i$ 
with $1 \leq i < M$, $B_i = B_{i-1}$; for even $i$ with 
$1 < i \leq M$, $A_i = A_{i-1}$; also $B_M = B_{M + 1} = B_{\mathrm{out}} B^\prime$, 
and $A_{M+1} = A_{\mathrm{out}} A^\prime$. In this way, after the application of $U_i$, Alice 
holds register $A_i$, Bob holds register $B_i$ and the communication register is $C_i$. In the 
case of an odd number of messages~$M$, the registers corresponding to $U_M, U_{M+1}$ are 
changed appropriately. We slightly abuse notation and also write $\Pi$ to denote the channel 
from $A_{\mathrm{in}} B_{\mathrm{in}}$ to $A_{\mathrm{out}} B_{\mathrm{out}}$ 
implemented by the protocol. That is, for any $\rho \in \D(A_{\mathrm{in}} B_{\mathrm{in}})$,
\begin{align*}
\Pi (\rho) \quad \eqdef \quad \Tra{A^\prime B^\prime} \left[ U_{M + 1} U_M \cdots U_2 U_1 (\rho \otimes \psi)\right]\,.
\end{align*}
The registers~$A^\prime $ and~$ B^\prime $ that are 
discarded by Alice and Bob, respectively, are two of the registers
at the end of the protocol.
As on the RHS above, we sometimes omit the superscript if the system is clear from context.

We restrict our attention to protocols with classical inputs~$XY$,
with~$A_\rin B_\rin$ initialized to~$XY$, and to
so-called ``safe protocols''. Safe protocols
only use~$A_\rin B_\rin$ as control registers. As explained in
Section~\ref{sec-qic}, this does not affect the results presented in
this article.

\suppress{
Lauri\`ere and Touchette~\cite{LT17-information-flow,LT17-information-flow-arxiv}, show that 
in any protocol~$\Pi$ with classical inputs, if
the two parties encode their inputs into a 
quantum register and modify them, the associated quantum information 
cost is at least as much as in a protocol~$\Pi'$ that makes a 
local copy of the inputs at 
the outset (into registers that are never touched again) and then 
simulates the original protocol~$\Pi$.
(See Section~\ref{sec-qic} for the definition of quantum information
cost we use in this work.)
}

We imagine that the joint classical input $XY$ is purified by a
register~$R$. 
We often partition the purifying register as~$R = R_X R_Y$,
indicating that the classical input~$XY$, distributed as~$\nu$, and
represented by the quantum state $\rho_\nu$,
\begin{align*}
 \rho_\nu^{XY} \quad \eqdef \quad \sum_{x, y} \nu (x, y) \; \kb{x}{x}^X \otimes \kb{y}{y}^Y \enspace,
\end{align*}
is purified as
\begin{align*}
 \ket{\rho_\nu} \quad \eqdef \quad \sum_{x, y} \sqrt{\nu (x, y)} \;
\ket{xxyy}^{XR_X Y R_Y} \enspace.
\end{align*}
We also use other partitions more appropriate for our purposes, 
corresponding to particular preparations of the inputs $X$ and $Y$. 

We define the quantum communication cost of $\Pi$ from Alice to Bob as 
\begin{align*}
\QCC_{\rA \to \rB} (\Pi) \quad
    \eqdef \quad   \sum_{i \in [0, (M-1)/2]} \log \size{C_{2i +1}} \,,
\end{align*}
and the quantum communication cost of $\Pi$ from Bob to Alice as
\begin{align*}
\QCC_{\rB \to \rA} (\Pi) \quad 
    \eqdef \quad \sum_{i \in [1, M/2]} \log \size{C_{2i}} \,,
\end{align*}
where for a register~$D$, the notation~$\size{D}$ stands for the dimension
of the state space associated with the register.
The total communication cost of the protocol is then the sum of these two quantities.

\begin{figure}
\begin{overpic}[width=1\textwidth]{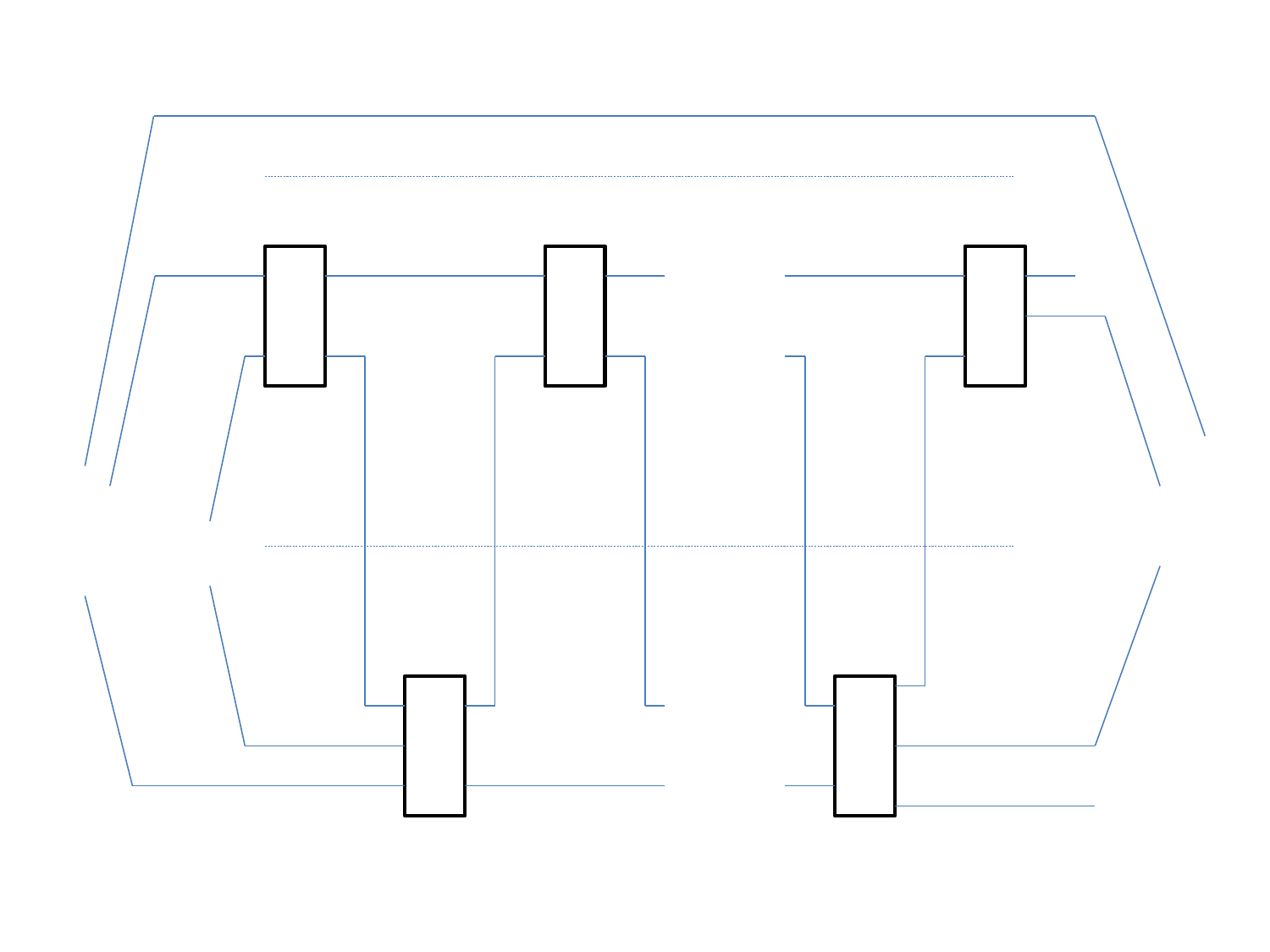}
  \put(0,65){Reference}
  \put(0,45){Alice}
  \put(0,22){Bob}
  \put(6,31){\footnotesize $\ket{\rho_\nu}$}
  \put(15,66.5){\footnotesize$R$}
  \put(15,54){\footnotesize$A_{\rin}$}
  \put(15,13.8){\footnotesize$B_{\rin}$}
  \put(15,44){\footnotesize$T^\rA$}
  \put(15,19){\footnotesize$T^\rB$}
  \put(22.1,49.5){\footnotesize$U_1$}
  \put(15,31){\footnotesize$\ket{\psi}$}
  \put(26.2,54){\footnotesize$A_1$}
  \put(26.2,48){\footnotesize$C_1$}
  \put(33,15.5){\footnotesize$U_2$}
  \put(37.2,54){\footnotesize$A_2$}
  \put(37.2,17.4){\footnotesize$C_2$}
  \put(37.2,13.7){\footnotesize$B_2$}
  \put(44.2,49.5){\footnotesize$U_3$}
  \put(48.2,54){\footnotesize$A_3$}
  \put(48.2,48){\footnotesize$C_3$}
  \put(48.2,13.7){\footnotesize$B_3$}
  \put(55,33){\footnotesize$\cdots$}
  \put(59.5,54){\footnotesize$A_{M-1}$}
  \put(59.5,48){\footnotesize$C_{M-1}$}
  \put(59.5,13.7){\footnotesize$B_{M-1}$}
  \put(66.5,15.5){\footnotesize$U_{M}$}
  \put(71.7,54){\footnotesize$A_M$}
  \put(73.5,23){\footnotesize$C_M$}
  \put(73.5,17.0){\footnotesize$B_{\rout}$}
  \put(73.5,12.0){\footnotesize$B^{\prime}$}
  \put(75.9,49.5){\footnotesize$U_{M+1}$}
  \put(82.3,54){\footnotesize$A^\prime$}
  \put(82.3,50.9){\footnotesize$A_{\rout}$}
  \put(92,33){\footnotesize$\Pi (\rho_\nu)$}
\end{overpic}
  \caption{Depiction of an interactive quantum protocol, adapted
from Ref.~\cite[Figure~1]{Touchette14}, the full version of
Ref.~\cite{Touchette15}.}
  \label{fig:int_mod}
\end{figure}


\subsection{Distance Measures}
\label{sec:distmeas}

In order to distinguish between quantum states, we use two related distance measures: trace distance and Bures distance.

\paragraph{Trace Distance.}

The trace distance between two quantum states $\rho$ and 
$\sigma$ in the same state space is denoted as $\| 
\rho - \sigma \|_1$ , where
\begin{align*}
	\| O \|_1 \quad \eqdef \quad \Tr{}{(O^\dagger O)^{\tfrac{1}{2}}}
\end{align*}
is 
the trace norm for operators. In operational terms,
the trace distance between the two states $\rho$ and $\sigma$ is 
four times the best possible bias with which we can distinguish 
between the two states, given a single copy of an unknown state out of the two.

We use the following properties of trace distance.
First, it is a metric, so it is symmetric in $\rho, \sigma$,
non-negative, evaluates to $0$ if and only if $\rho = \sigma$,
and it satisfies the triangle inequality. Moreover, it is monotone under
 the action of channels: for any $\rho_1, \rho_2 \in \D (A)$ and channel $\N^{A \rightarrow B}$ from system $A$ to system $B$,
\begin{align*}
	\| \N (\rho_1) - \N (\rho_2) \|_1 \quad \leq \quad \| \rho_1 - \rho_2 \|_1
\enspace.
\end{align*}
For isometries, the inequality is tight, a property called isometric
 invariance of the trace distance. Hence, for any $\rho_1$, $\rho_2 \in \D(A)$ and any isometry $U^{A \rightarrow B}$, we have
\begin{align*}
	\| U (\rho_1) - U (\rho_2) \|_1 \quad = \quad \| \rho_1 - \rho_2 \|_1
\enspace.
\end{align*}
Trace distance obeys a joint linearity property: 
for a classical system $X$ and two states
\begin{align}
\begin{split}
\rho_1^{XA} & \quad \eqdef \quad \sum_x p_X (x) \, \kb{x}{x}^X \otimes \rho_{1, x}^A \enspace, \qquad \text{and} \\
\rho_2^{XA} & \quad \eqdef \quad \sum_x p_X (x) \, \kb{x}{x}^X \otimes \rho_{2, x}^A \enspace,
\label{eq-cqstates}
\end{split}
\end{align}
we have
\begin{align}
\label{eq-tr-joint-linearity}
	\| \rho_1 - \rho_2 \|_1 \quad = \quad \sum_x p_X (x) \; \| \rho_{1,x} - \rho_{2, x}  \|_1
\enspace.
\end{align}

\paragraph{Bures Distance.}

Bures distance $\fh$ is a fidelity based distance measure, defined for $\rho, \sigma \in \D (A)$ as
\begin{align*}
\fh  \!\left(\rho, \sigma \right) \quad \eqdef \quad \sqrt{1 - \rF (\rho, \sigma)}\ ,
\end{align*}
where fidelity $\rF$ is defined as $\rF (\rho, \sigma) \eqdef \| \sqrt{\rho} \sqrt{\sigma} \|_1$.

We use the following properties of Bures distance. First, it is a metric, so it is symmetric in $\rho, \sigma$, 
non-negative, evaluates to $0$ if and only if $\rho = \sigma$, and satisfies the triangle inequality. Moreover, it is monotone
under
 the action of a channel: for any $\rho_1$, $\rho_2 \in \D (A)$ and quantum
channel $\N^{A \rightarrow B}$,
\begin{align*}
	\fh \!\left( \N (\rho_1) , \N (\rho_2) \right) \quad \leq \quad
\fh \!\left( \rho_1,  \rho_2  \right)
\enspace.
\end{align*}
For isometries, the inequality is tight, a property called isometric
 invariance of the Bures distance.

It is sometimes convenient to work with the square of the Bures distance.
In particular, the square obeys a joint linearity property: 
for a classical system $X$ and two states $\rho_1^{XA}, \rho_2^{XA}$ defined as in Eq.~(\ref{eq-cqstates}),
\begin{align}
\label{eq-bures-joint-linearity}
	\fh^2  \!\left( \rho_1 , \rho_2  \right) \quad = \quad \sum_x p_X (x) \,
\fh^2  \!\left( \rho_{1,x} , \rho_{2, x}   \right)
\enspace.
\end{align}
It also satisfies a weaker version of the triangle inequality: for any $\rho_1$, $\rho_2$ and $\sigma \in \D(A)$,
\begin{align}
\label{eq-weak-triangle}
\fh^2  \!\left(\rho_1 , \rho_2 \right) \quad \leq \quad 2 \fh^2
\!\left(\rho_1, \sigma \right) + 2 \fh^2  \!\left(\sigma , \rho_2 \right) \enspace.
\end{align}

\paragraph{Local Transition Lemma.}

The following lemma, a direct consequence of the Uhlmann theorem, is
called the local transition lemma~\cite{KNTZ07},
especially when expressed 
in terms of other metrics.

\begin{lemma}\label{lem:loctr}
Let $\rho_1$, $\rho_2 \in \D (A)$ have purifications $\rho_1^{A R_1}$, $\rho_2^{A R_2}$, with $|R_1| \leq |R_2|$. Then, there exists an isometry $V^{R_1 \rightarrow R_2}$ such that
\begin{align*}
\fh \big(\rho_1^{A}, \rho_2^{A}\big) \quad = \quad \fh \Big(V \big(\rho_1^{A R_1}\big), \rho_2^{A R_2} \Big)\ .
\end{align*}
\end{lemma}

Bures distance is related to trace distance through a generalization of the Fuchs-van de Graaf inequalities~\cite{FuchsG99}:
for any $\rho_1$, $\rho_2 \in \D (A)$ , it holds that
\begin{align}
	\label{eq:fvdg}
	\fh^2  \!\left(\rho_1, \rho_2 \right)
            \quad \leq \quad \frac{1}{2} \; \| \rho_1 - \rho_2 \|_1
            \quad \leq \quad \sqrt{2} \fh \!\left(\rho_1, \rho_2
\right) \enspace.
\end{align}

\subsection{Information Measures}
\label{app:infomeas}

In order to quantify the information content of a quantum state,
we use a basic measure, von Neumann entropy, defined as
\begin{align*}
	\rH(A)_\rho \quad \eqdef \quad - \Tr{}{\rho \log \rho}
\end{align*}
for any state $\rho \in \D(A)$.
Here, we follow the convention that $0 \log 0 = 0$, which is
justified by a continuity argument. The logarithm is in base $2$. 
Note that $\rH$ is invariant under isometries applied on $\rho$.
If the 
state in question is clear from the context, we may omit the subscript. 
We also note that if system $A$ is classical, then von Neumann entropy
reduces to Shannon entropy. 

For a state $\rho^{ABC} \in \D(A  B  C)$, the mutual information between
registers~$A,B$ is defined as
\begin{align*}
	\rI(A \!:\! B )_\rho \quad \eqdef \quad \rH(A)  + \rH(B) - \rH(AB)
\enspace,
\end{align*}
and the conditional mutual information between them, given~$C$, as
\begin{align*}
	\rI(A \!:\! B \,|\,  C )_\rho \quad \eqdef \quad \rI (A \!:\! BC) -
\rI(A  \!:\!  C)
\enspace.
\end{align*}
If $X$ is a classical system,
	$\rI(X  \!:\! B)$
is also called the Holevo information.

Mutual information and conditional mutual information are symmetric
in~$A,B$, and invariant under a local isometry applied to $A, B$ or $C$.
Since all purifications of a state are equivalent up to an isometry on the purification registers, we have that for any two pure states
$\ket{\phi}^{ABC R^\prime}$ and $\ket{\psi}^{ABC R}$ such that 
$\phi^{ABC} = \psi^{ABC}$,
\begin{align*}
	\rI(C \!:\!R'\,|\, B)_{\phi} \quad = \quad \rI(C \!:\!R\,|\, B)_{\psi} \enspace.
\end{align*}
For any state $\rho \in \D(A  B  C )$, we have the bounds
\begin{align*}
	0 \quad \leq \quad \rH(A) \quad \leq \quad \log |A| \enspace, \\
	0 \quad \leq \quad \rI(A \!:\! B \,|\,  C) \quad \leq  \quad 2 \; \rH(A)
\enspace.
\end{align*}
For a multipartite quantum system $ABCD$, conditional mutual information 
satisfies a chain rule: for any $\rho \in \D(A  B  C  D)$,
\begin{align*}
	\rI (AB  \!:\! C \,|\,  D) \quad = \quad \rI (A \!:\! C \,|\,
D) + \rI(B \!:\! C \,|\,  A D)
\enspace.
\end{align*}
For any product state $\rho^{A_1 B_1 A_2 B_2} \eqdef
\rho_1^{A_1 B_1} \otimes \rho_2^{A_2 B_2}$, entropy is additive across
the bi-partition, so that, for example,
\begin{align*}
	\rH(A_1 A_2) \quad = \quad \rH(A_1) + \rH(A_2) \enspace,
\end{align*}
and the conditional mutual information between product systems vanishes:
\begin{align*}
	\rI(A_1  \!:\! A_2 \,|\,  B_1 B_2 ) \quad = \quad 0 \enspace.
\end{align*}
Two important properties of the conditional mutual information are non-negativity and the data processing inequality, both equivalent to a deep result in quantum information theory known as strong subadditivity~\cite{LR73}.
For any state $\rho \in \D (A  B  C)$, channel $\N^{B \rightarrow B^\prime} $, and state $\sigma \eqdef \N (\rho)$,
we have
\begin{align*}
	\rI (A \!:\! B \,|\,  C)_\rho  & \quad \geq \quad 0, \\
	\rI (A \!:\! B \,|\,  C)_\rho  & \quad \geq \quad \rI (A \!:\! B^\prime \,|\,  C)_\sigma
\enspace.
\end{align*}
For classical systems, conditioning is equivalent to taking an average:
for any state
\[
\rho^{ABCX} \quad \eqdef \quad \sum_x p_X(x) \, \kb{x}{x}^X \otimes \rho_x^{ABC}
\]
with a classical system $X$ and some states $\rho_x \in \D (A  B  C)$,
\begin{align*}
\label{eq:cqmicond}
	\rI (A \!:\! B \,|\,  C X )_\rho & \quad = \quad \sum_x p_X (x) \rI(A \!:\! B \,|\,  C)_{\rho_x}
\enspace.
\end{align*}

\paragraph{Average Encoding Theorem.} The following lemma, known as the Average Encoding Theorem~\cite{KNTZ07,JainRS03b}, formalizes the intuition that if a classical and a quantum registers are weakly correlated, then they are nearly independent.

\begin{lemma}
\label{lem:avenc}
For any $\rho^{XA} \eqdef \sum_x p_X(x) \, \kb{x}{x}^X \otimes \rho_x^{A}$ with a classical system $X$ and
 states $\rho_x \in \D (A )$,
\begin{align*}
	\sum_x p_X (x) \fh^2  \!\left(\rho_x^A, \rho^A \right) &
\quad \leq \quad  \rI(X \!:\! A )_{\rho}
\enspace.
\end{align*}

\end{lemma}

\paragraph{Fawzi-Renner inequality.}

We use the following breakthrough result by Fawzi and Renner~\cite{FawziR15}.
It provides a lower bound on the quantum conditional mutual information in 
terms of the fidelity for the optimal recovery map.

\begin{lemma}
\label{lem:frineq}
For any tripartite quantum state $\rho^{ACR}$, there exists a recovery 
map $\R^{A \rightarrow AC}$ from register $A$ to registers $AC$ satisfying
\begin{align*}
\rI(C \!:\!R \,|\, A) \quad \geq \quad - 2 \, \log_2 \rF (\rho^{ACR}, \R (\rho^{AR})) \enspace.
\end{align*}
In particular, it follows that 
\begin{align*}
\rI(C \!:\!R \,|\, A) \quad \geq \quad \fh^2  \!\left(\rho^{ACR}, \R
(\rho^{AR}) \right) \enspace.
\end{align*}
\end{lemma}

\subsection{Quantum Information Complexity}
\label{sec-qic}

We rely on the notion of quantum information cost of a two-party
communication protocol introduced by
Touchette~\cite{Touchette15}. We follow the notation and conventions
associated with a two-party quantum communication protocol introduced in
Section~\ref{sec-qcc}, and restrict ourselves to protocols
with classical inputs~$XY$ distributed as~$\nu$.

Quantum information cost is defined in terms of the purifying 
register $R$, but is independent of the choice of purification.
Given the asymmetric nature of the Augmented Index function, we consider 
the quantum information cost of messages from Alice to Bob and the ones 
from Bob to Alice separately. 
Such an asymmetric notion of quantum information cost was previously
considered in Refs.~\cite{KerenidisLLR16, LT17-information-flow, LT17-information-flow-arxiv}. 
\begin{definition}
Given a quantum protocol $\Pi$ with classical inputs distributed as~$\nu$,
the \emph{quantum information cost} (of the messages)
from Alice to Bob is defined as
\begin{align*}
\QIC_{\rA \rightarrow \rB} (\Pi, \nu) \quad \eqdef \quad &
\sum_{i~\mathrm{odd}} \rI (R  \!:\! C_i \,|\,  B_i) \enspace,
\end{align*}
and the \emph{quantum information cost} (of the messages) from Bob to Alice is defined as
\begin{align*}
\QIC_{\rB \rightarrow \rA} (\Pi, \nu) \quad \eqdef \quad &
\sum_{i~\mathrm{even}} \rI (R  \!:\! C_i \,|\,  A_i) \enspace.
\end{align*}
\end{definition}

It is immediate that quantum information cost is bounded above by 
quantum communication.
\begin{remark}
For any quantum protocol $\Pi$ with classical inputs distributed 
as $\nu$, the following hold:
\begin{align*}
\QIC_{\rA \rightarrow \rB} (\Pi, \nu) &
    \quad \leq \quad 2 \; \QCC_{A \rightarrow B} (\Pi) \enspace, \qquad \text{and} \\
\QIC_{\rB \rightarrow \rA} (\Pi, \nu) &
    \quad \leq \quad 2 \; \QCC_{B \rightarrow A} (\Pi) \enspace.
\end{align*}
\end{remark}
As a result, we may bound 
quantum communication complexity of a protocol from below by analysing
its information cost.

We further restrict ourselves to ``safe protocols'', in which the
registers~$A_\rin,B_\rin$ are only used as control registers in the
local isometries. This restriction does not affect the results in this
article, for the following reason. Let~$\Pi$ be any protocol with
classical inputs distributed as~$\nu$, in which the two parties may 
apply arbitrary isometries to their quantum registers. In particular,
these registers include~$A_\rin,B_\rin$ which are initialized to the input.
Let~$\Pi'$ be the protocol with the same registers as~$\Pi$ and two
additional quantum registers~$A'_\rin,B'_\rin$ of the same sizes 
as~$A_\rin,B_\rin \,$, respectively. In the protocol~$\Pi'$, the two
parties each make a coherent local copy of their inputs 
into~$A'_\rin,B'_\rin \,$, respectively, at the outset. That is, Alice applies an isometry that maps
\[
\ket{x}^{A_\rin} \quad \mapsto \quad \ket{x}^{A'_\rin} \ket{x}^{A_\rin} \enspace,
\]
for each possible classical input~$x$, and Bob applies a similar isometry to~$B_\rin$.
The registers~$A'_\rin,B'_\rin$ are never touched hereafter,
and the two parties simulate the original protocol~$\Pi$ on the
remaining registers. As Lauri\`ere and Touchette~\cite[Proposition~9]{LT17-information-flow-arxiv} show, the quantum information cost of~$\Pi$ is at least as much as that of the protocol~$\Pi'$:
\begin{align*}
\QIC_{\rA \rightarrow \rB}(\Pi',\nu)
    \quad & \leq \quad \QIC_{\rA \rightarrow \rB}(\Pi,\nu) \enspace,
    \qquad \textrm{and} \\
\QIC_{\rB \rightarrow \rA}(\Pi',\nu)
    \quad & \leq \quad \QIC_{\rB \rightarrow \rA}(\Pi,\nu) \enspace.
\end{align*}
Thus, the quantum information cost trade-off we show for safe protocols 
holds for arbitrary protocols as well.

We use another result due to Lauri\`ere and 
Touchette~\cite{LT17-information-flow,LT17-information-flow-arxiv}. The result states
that the total gain in (conditional) information by a party 
over all the messages is precisely the net (conditional) information 
gain at the end of the protocol. It allows us to keep track of the flow of 
information during interactive quantum protocols. 
For completeness, a proof is provided in Appendix~\ref{app-ifl}.
\begin{lemma}[Information Flow Lemma]
\label{thm-inf-flow}
	Given a protocol $\Pi$, an input state $\rho$ with purifying
register $R$ with 
arbitrary decompositions $R = R_a^\rA R_b^\rA R_c^\rA = R_a^\rB R_b^\rB
R_c^\rB$,
the following hold:
\begin{align*}
\MoveEqLeft \sum_{i \geq 0} \rI (R_a^\rB  \!:\! C_{2i + 1} \,|\,  R_ b^\rB B_{2i + 1} ) -
\sum_{i \geq 1} \rI (R_a^\rB  \!:\! C_{2i} \,|\,  R_b^\rB B_{2i}) \qquad \\
    & = \quad \rI (R_a^\rB  \!:\! B_{\rout} B^\prime \,|\,  R_b^\rB)
          - \rI (R_a^\rB  \!:\! B_{\rin} \,|\,  R_b^\rB) \enspace, \quad
          \textrm{and} \\
 & \quad \\
\lefteqn{ \sum_{i \geq 0} \rI (R_a^\rA  \!:\! C_{2i + 2} \,|\,  R_ b^\rA A_{2i + 2} ) -
\sum_{i \geq 0} \rI (R_a^\rA  \!:\! C_{2i + 1} \,|\,  R_b^\rA A_{2i + 1}) }
\qquad \\
    & = \quad \rI (R_a^\rA  \!:\! A_{\rout} A^\prime \,|\,  R_b^\rA)
          - \rI (R_a^\rA  \!:\! A_{\rin} \,|\,  R_b^\rA) \enspace.
\end{align*}
\end{lemma}

\subsection{Quantum Streaming Algorithms}
\label{sec-qsa}

We refer the reader to the text~\cite{Muthukrishnan05} for an
introduction to classical streaming algorithms. Quantum streaming
algorithms are similarly defined, with restricted access to the input, and
with limited workspace. 

In more detail, an input~$x\in\Sigma^n \,$, where~$\Sigma$
is some alphabet, arrives as a \emph{data stream\/}, i.e., letter by
letter in the order~$x_1, x_2, \dotsc, x_n$. An algorithm is said to
make a \emph{pass\/} on the input, when it reads the data stream once
in this order, processing it as described next.
For an integer~$T \ge 1$, a~$T$-pass (unidirectional) \emph{quantum
streaming algorithm\/}~$\cA$ with space~$s(n)$ and time~$t(n)$ is a
collection of quantum channels~$\set{ \cA_{i\sigma} : i \in [T],
\sigma \in \Sigma}$. Each operator~$\cA_{i\sigma}$ is a channel defined on a
register of~$s(n)$-qubits, and can be implemented
by a uniform family of circuits of size at most~$t(n)$. On input stream~$x \in
\Sigma^n \,$,
\begin{enumerate}
\item
The algorithm starts with a register~$W$ of~$s(n)$ qubits, all 
initialized to a fixed state, say~$\ket{0}$.

\item
$\cA$ performs~$T$ sequential passes, $i = 1, \dotsc, T$, 
on~$x$ in the order~$x_1, x_2, \dotsc, x_n$.

\item
In the~$i$th pass, when symbol~$\sigma$ is read, channel~$\cA_{i\sigma}$
is applied to~$W$.

\item
The output of the algorithm is the state in a designated
sub-register~$W_\rout$ of~$W$, at the end of the~$T$ passes.
\end{enumerate}
We may allow for some pre-processing before the input is read, 
and some post-processing at the end of the~$T$ passes, each with 
time complexity different from~$t(n)$. As our work applies to streaming
algorithms with any time complexity, we do not consider this
refinement.

The probability of correctness of a streaming algorithm is defined in
the standard way.  If we wish to compute a family of Boolean 
functions~$g_n : \Sigma^n \rightarrow \set{0,1}$, the output 
register~$W_\rout$ consists of a single qubit. On input~$x$, 
let~$\cA(x)$ denote the random variable corresponding to the outcome when 
the output register is measured in the standard
basis. We say~$\cA$ computes~$g_n$ with (worst-case)
error~$\epsilon \in [0,1/2]$
if for all~$x$, $\Pr[\cA(x) = g_n(x)] \ge 1 - \epsilon$.

In general, the implementation of a quantum channel used by a
streaming algorithm with unitary operations involves one-time use of 
ancillary qubits (initialized to a fixed, known quantum state,
say~$\ket{0}$). These ancillary qubits are in addition to the~$s(n)$-qubit
register that is maintained by the algorithm. Fresh qubits may be an 
expensive resource in practice, for example, in NMR implementations, and
one may argue that they be included in the space complexity of the
algorithm. The lack of ancillary qubits
severely restricts the kind of computations space-bounded algorithms can
perform; see, for example, Ref.~\cite{ANTV02}.
We choose the definition above so as to present the results
we derive in the strongest possible model. Thus, the results
also apply to implementations in which the ``flying qubits'' needed for
implementing non-unitary quantum channels are relatively easy
to prepare.

In the same vein, we may provide a quantum streaming algorithm arbitrary 
read-only access to a sequence of random bits. In other words,
we may also provide the algorithm with a register~$S$ of size at
most~$nT t(n)$ initialized to random bits from some distribution. Each quantum 
channel~$\cA_{i\sigma}$ now operates on registers~$SW$, while 
using~$S$ only as a control register.  The bounds we prove hold in 
this model as well.

\section{Reduction from Augmented Index to \dyck(2)}
\label{sec-reductions}

The connection between low-information protocols for Augmented Index
and streaming algorithms for \dyck(2) contains two steps. The first is a
reduction from an intermediate multi-party communication problem \Asc, and the
second is the relationship of the latter with Augmented Index.

\subsection{Reduction from \Asc\ to \dyck(2)}
\label{sec-ascension}

In this section, we describe the connection between multi-party
quantum communication protocols for the 
problem \Asc$(m,n)$, and quantum streaming algorithms for \dyck(2).
The reduction is an immediate generalization of
the one in the classical case discovered by Magniez, Mathieu,
and Nayak~\cite{MagniezMN14}, which also works with appropriate 
modifications for multi-pass classical streaming 
algorithms~\cite{ChakrabartiCKM10,JainN14}. For the sake of 
completeness, we describe the reduction below.

Multi-party quantum communication
protocols involving point-to-point communication may be defined as in
the two-party case. In particular, we follow the convention that the initial joint state of all the parties is pure, and the local operations are isometries. So the joint state is pure throughout.
As it is straightforward, and detracts from the thrust of this section, we omit a formal definition.

Let~$m,n$ be positive integers.
The~$(2m)$-party communication problem~\Asc$(m,n)$ consists of computing
the logical OR of~$m$ independent instances of~$f_n \,$, the Augmented
Index function.
Suppose we denote the~$2m$ parties by~$\sA_1,\sA_2,\dotsc,\sA_m$ and
$\sB_1,\sB_2,\ldots,\sB_m$. Player~$\sA_i$ is given~$x^i\in\{0,1\}^n$,
player~$\sB_i$ is given $k^i\in [n]$, a bit~$z^i$, and the
prefix~$x^i[1, k^i-1]$ of $x^i$. Let $\mathbf{x} \eqdef (x^1,x^2,\ldots,x^m)$,
$\mathbf{k} \eqdef (k^1,k^2,\ldots,k^m)$, and
$\mathbf{z} \eqdef (z^1,z^2,\ldots,z^m)$.
The goal of the communication protocol is to compute
\[
F_{m,n}(\mathbf{x},\mathbf{k},\mathbf{z})
    \quad \eqdef \quad \bigvee_{i=1}^m f_n(x^i,k^i,z^i)
    \quad = \quad \bigvee_{i=1}^m (x^i[k^i] \xor z^i) \enspace,
\]
which is $0$ if $x^i[k^i]=z^i$ for all $i$, and $1$ otherwise.

The communication between the~$2m$ parties is required to be~$T$
sequential
iterations of communication in the following order, for some positive integer~$T$:
\begin{align}
\nonumber
\sA_1 \rightarrow \sB_1 \rightarrow
\sA_2 \rightarrow \sB_2 \rightarrow
\dotsb
\sA_m \rightarrow \sB_m \\
\label{eqn-commn}
\qquad \rightarrow \sA_m \rightarrow \sA_{m-1} \rightarrow
\dotsb \rightarrow \sA_2 \rightarrow \sA_1 \enspace.
\end{align}
In other words, for~$t = 1, 2, \dotsc, T$,
\begin{itemize}
\item[\ --] for $i$ from $1$ to $m-1$, player~$\sA_i$ sends a register
  $C_{\sA_i, t}$ to $\sB_i$, then $\sB_i$ sends register~$C_{\sB_i,t}$ to
$\sA_{i+1}$,
\item[\ --] $\sA_m$ sends register $C_{\sA_m, t}$ to $B_m$, then
  $B_m$ sends register $C_{B_m,t}$ to $\sA_{m}$,
\item[\ --] for $i$ from $m$ down to 2, $\sA_i$ sends register
  $C'_{\sA_i,t}$ to $\sA_{i-1}$.
\end{itemize}
At the end of the~$T$ iterations, $\sA_1$ computes the output.

There is a bijection between instances of \Asc$(m,n)$ and a subset of
instances of \dyck(2) that we describe next. For any string $z \eqdef z_1
\dotsb z_n \in\set{a,b}^n$, let $\overline{z}$
denote the matching string~$\overline{z_n} \, \overline{z_{n-1}} \,
\dotsb \overline{z_1}$
corresponding to $z$. Let~$z[i,j]$ denote the substring~$z_i z_{i+1}
\dotsb z_j$ if~$1 \leq i \leq j \leq n$, and the empty string~$\epsilon$
otherwise. We abbreviate~$z[i,i]$ as~$z[i]$ if~$1 \le i \le n$.
Consider a string
\begin{align}
w  \quad \eqdef \quad   (x^1 \, \overline{y^1} \, \overline{z^1} \, z^1 \, y^1) \,
        (x^2 \, \overline{y^2} \, \overline{z^2} \, z^2 \, y^2)
       \dotsb \ (x^m \, \overline{y^m} \, \overline{z^m} \, z^m \, y^m) \,
        \overline{x^m}  \ \dotsb\ \overline{x^{2}} \ \overline{x^1}
        \enspace,
\label{eqn-instance}
\end{align}
where for every $i$, $x^i\in \{ a,b\}^n$, and~$y^i$ is a suffix
of~$x^i$, i.e., $y^i={x^i[n-k^i+2,n]}$ for some
$k^i\in\{1,2,\ldots,n\}$, and $z^i\in \{ a,b\}$. The string $w$ is in
$\dyck(2)$ if and only if, for every $i$, $z^i=x^i[n-k^i+1]$. Note that
such instances have length in the interval~$[2m(n+1), 4mn]$.

The bijection between instances of \Asc$(m,n)$ and \dyck(2) arises from a
partition of the string~$w$ amongst the~$2m$ players: player~$\sA_i$ is
given~$x^i$ (and therefore~$\overline{x^i}$), and player~$\sB_i$ is
given~$y^i,z^i$ (and therefore~$\overline{y^i}, \overline{z^i}$).
See Figure~\ref{fig-asc} for a pictorial representation of the
partition.
\begin{figure}[!t]
\centering
\includegraphics[width=350pt]{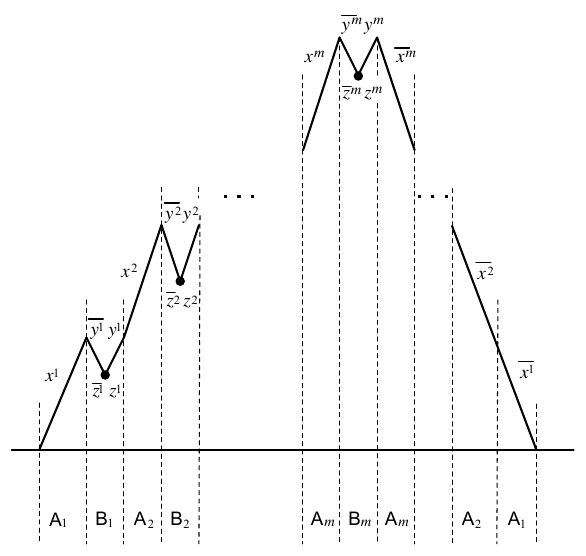}
\caption{An instance of the form described in~(\ref{eqn-instance}), as
depicted in~\cite{MagniezMN14,JainN14}. A line
segment with positive slope denotes a string over~$\set{a,b}$, and a
segment with negative slope denotes a string over~$\set{\overline{a},
\overline{b}}$. A solid dot depicts a pair of the form~$\overline{z} z$
for some~$z \in \set{a,b}$. The entire string is distributed
amongst~$2m$
players~$\sA_1, \sB_1, \sA_2, \sB_2,  \dotsc, \sA_m, \sB_m$ in
a communication protocol for~$\Asc(m,n)$ as shown.}
\label{fig-asc}
\end{figure}
For ease of notation, the strings~$x^i$ in~$\Asc(m,n)$ are taken to be the
ones in $\dyck(2)$ with the bits \emph{in reverse order\/}. This converts
the suffixes~$y^i$ into prefixes of the same length.

As a consequence of the bijection above, any quantum streaming algorithm
for \dyck(2) results in a quantum protocol for \Asc$(m,n)$, as stated in 
the following lemma.
\begin{lemma}
\label{lem:reddycktoasc}
For any $\epsilon \in [0,1/2]$, $ n, m \in \mathbb{N}$, 
and for any $\epsilon$-error (unidirectional) $T$-pass quantum 
streaming algorithm $\cA$ for \dyck(2) that on instances of
size $N \in \Theta(mn)$ uses~$s(N)$ qubits of memory, 
there exists an $\epsilon$-error, $T$-round
sequential $(2m)$-party quantum communication
protocol for \Asc$(m,n)$ in which each message is of length~$s(N)$.
The protocol may use public randomness, but does not use shared 
entanglement between any of the parties. Moreover, the local operations
of the parties are memory-less, i.e., do not require access to
the qubits used in generating the previous messages sent by that party.
\end{lemma}
\begin{proof}
Any random sequence of bits used by the streaming algorithm is provided
as shared randomness to all the~$2m$ parties in the 
communication protocol for \Asc$(m,n)$. Each input for the communication
problem corresponds to an instance of \dyck(2), as described above.
In each of the~$T$ iterations, a player simulates the quantum streaming
algorithm on appropriate part of the input for \dyck(2), and sends the 
length~$s(N)$ workspace to the next player in the sequence. 
(If needed, non-unitary quantum operations may be replaced with an
isometry, as follows from the Stinespring Representation 
theorem~\cite{Watrous18}.)
The output of the protocol is the output of the algorithm, and is contained
in the register held by the final party~$\sA_1$.
\end{proof}

\subsection{Reduction from Augmented Index to \Asc}
\label{sec-ai2asc}

Recall that \Asc$(m,n)$ is composed of~$m$ instances of Augmented 
Index on~$n$ bits. Magniez, Mathieu, and Nayak~\cite{MagniezMN14} showed how
we may derive a low-information classical protocol for Augmented
Index~$f_n$ from one for \Asc$(m,n)$ through a direct sum argument (see
Refs.~\cite{ChakrabartiCKM10,JainN14} for the details of its working in the multi-pass case).
This is not as straightforward to execute as it might first appear; it
entails deriving a sequence of protocols for Augmented Index in which 
the communication from Alice to Bob corresponds to messages from 
different parties in the original multi-party protocol.
We show that the same kind of construction, suitably adapted to the notion 
of quantum information cost we use, also works in the quantum case.

Let~$\mu_0$ be the uniform distribution on the~$0$-inputs of the
Augmented Index function~$f_n$. If $X$ is a uniformly random~$n$-bit
string, $K$ is a uniformly random index from~$[n]$ independent of~$X$,
and the random variable~$B$ is set as $B = X_K$, the joint 
distribution of~$X,K,X[1,K-1],B$ is~$\mu_0$. We denote the
random variables~$K,X[1,K-1],B$ given as input to Bob by~$Y$. Since~$X_K
= B$ under this distribution, we abbreviate Bob's input as~$K,X[1,K]$.
Let~$\mu$ be the uniform distribution over all inputs.
Under this distribution, the bit~$B$ is uniformly 
random, independent of~$XK$, while~$XK$ are as above.

\begin{lemma}
\label{lem:redaugind}
Suppose~$\epsilon \in [0,1/2]$, $ n, m \in \mathbb{N}$ and
that there is an~$\epsilon$-error, $T$-round sequential quantum protocol~$\Pi_{\asc}$ 
for \Asc$(m,n)$, that is memory-less, does not have shared 
entanglement between any of the parties (but might use public randomness), 
and only has messages of length 
at most~$s$ (cf.\ Lemma~\ref{lem:reddycktoasc}).
Then, there exists an $\epsilon$-error, $2T$-message, two-party quantum
protocol~$\Pi_{\ai}$ for the Augmented Index function~$f_n$ that satisfies
\begin{align*}
\QIC_{\rA \rightarrow \rB} (\Pi_\ai, \mu_0)
    & \quad \leq \quad 2 s T \enspace, \\
\QIC_{\rB \rightarrow \rA} (\Pi_\ai, \mu_0) 
    & \quad \leq \quad 2 s T/ m \enspace.
\end{align*}
\end{lemma}

\begin{proof}
Starting from the~$(2m)$-party protocol~$\Pi_\asc$ for \Asc$(m,n)$, we 
construct a protocol~$\Pi_j$ for~$f_n \,$, for each~$j \in [m]$. 

Fix one such~$j$.
Suppose Alice and Bob get inputs $x$ and $y$, respectively, 
where~$y \eqdef(k, x [1, k-1], b)$.
They embed these into an instance of \Asc$(m,n)$: they
set $x_j = x$, and $y_j = y$. They sample the inputs for the
remaining~$m-1$ coordinates independently, according to $\mu_0$. 
Let $X_i Y_i$, with~$Y_i \eqdef (K_i, X_i[1, K_i])$, be registers 
corresponding to inputs drawn from $\mu_0$ in 
coordinate $i$. Let $R_i$ be a purification register for these, which 
we may decompose as $R_i^X R_i^Y$, denoting the standard purification
of the $X_i Y_i$ registers. Let $S_\rA S_\rB$ be registers initialized
to~$\sum_s \sqrt{\nu_s} \ket{ss}$, so as to simulate the
public random string~$S \sim \nu$ used in the protocol~$\Pi_\asc$.

For each $i \not= j$, we give $X_i$ to Alice, and $(K_i, X_i [1, K_i])$ to Bob. For $i >j$, 
we give $R_i$ to Bob, and for $i < j$, we give $R_i$ to Alice. Then Alice and Bob simulate 
the roles of the~$2m$ parties~$(\sA_i, \sB_i)_{i \in [m]}$ in the following way for each
of the $T$ rounds in~$\Pi_\asc$. For~$t = 1, 2, \dotsc, T$:
\begin{enumerate} 
\item Alice simulates $\sA_1 \rightarrow \sB_1 \rightarrow \sA_2
\rightarrow \cdots \rightarrow \sA_j$, accessing the inputs for~$\sB_i$,
for~$i < j$, in the register~$R_i$. We denote the ancillary register she
uses to implement~$\sA_1$'s local isometry via unitary operators by~$D_{t1}$, and for all~$i \in [2,j]$, the ancillary registers she uses for~$\sB_{i-1}$ and~$\sA_i$ together
by~$D_{ti}$.

\item Alice transmits the message from $\sA_j$ to $\sB_j$ to Bob. 

\item Bob simulates $\sB_j \rightarrow \sA_{j+1} \rightarrow \cdots
\rightarrow \sB_m$, accessing the input for~$\sA_i$, for~$i > j$, in the
register~$R_i$.
For all~$i$ such that~$j \le i < m$ we denote the ancillary registers 
Bob uses for implementing~$\sB_i$ and~$\sA_{i+1}$'s local isometry together 
by~$D_{t(i+1)}$, and the ancillary register he uses for $\sB_m$ by~$D_{t(m+1)}$.

\item Bob transmits the message from $\sB_m$ to $\sA_m$ to Alice.

\item Alice simulates $\sA_m \rightarrow \sA_{m-1} \rightarrow \cdots
\rightarrow \sA_1$.
We denote the ancillary registers Alice uses for implementing the
local isometries of~$\sA_m, \sA_{m-1}, \dotsc, \sA_2$ together by~$E_{t}$.

\end{enumerate}
We let~$E_0$ denote a dummy register initialized to a fixed state,
say~$\ket{0}$.

Since the inputs for Augmented Index for~$i \neq j$ are distributed
according to $\mu_0$, the protocol~$\Pi_j$ computes Augmented Index
for the instance~$(x,y)$ with error at most~$\epsilon$.

The quantum information cost from Alice to Bob $\QIC_{\rA \rightarrow
\rB}(\Pi_j,\lambda)$ is bounded by~$2sT$, for any distribution~$\lambda$
over the inputs, as each of her~$T$ messages has
at most~$s$ qubits.

The bound on quantum information cost from Bob to Alice arises from the
following direct sum result. Suppose that the inputs for the
protocol~$\Pi_j$ for Augmented Index are drawn from the distribution~$\mu_0$. 
Denote these inputs by~$X_j Y_j$, with~$Y_j \eqdef (K_j, X_j [1, K_j])$, and
the corresponding purification register by~$R_j$.
We are interested in the quantum information cost
$\QIC_{\rB \rightarrow \rA} (\Pi_j, \mu_0)$.

For~$t \in [T]$, let~$C_t$ denote the~$t$th message from Bob to Alice in
the protocol~$\Pi_j$.
At the time Alice receives message~$C_t$, her other registers are~$X_1
\dotsb X_m$, $S_\rA$, $R_1 \dotsb R_{j-1}$, $(E_{r-1} D_{r1} D_{r2}
\dotsb D_{rj})_{r \in [t]}$.
Note that the corresponding state~$\rho_t$ at that point on registers
\[
X_1 \dotsb X_m \; S_\rA \; (E_{r-1} D_{r1} D_{r2} \dotsb D_{rm} D_{r (m+1)} )_{r \in
[t]} \; R_1 \dotsb R_m \; C_t
\]
is the same for all derived protocols~$\Pi_j$, as all of them
simulate~$\Pi_\asc$ on the same input distribution, namely~$\mu_0^{\otimes m}$,
using the above registers.

We have
\begin{align*}
\MoveEqLeft \QIC_{\rB \rightarrow \rA} (\Pi_j, \mu_0)  \\
    & = \quad \sum_{t \in [T]} \rI(R_j  \!:\! C_t \;|\; X_1 \dotsb X_m S_\rA
              (E_{r-1} D_{r1} D_{r2} \dotsb D_{rj})_{r \in [t]}
              R_1 \dotsb R_{j-1} )_{\rho_t} \\
    & \le \quad \sum_{t \in [T]} \rI(R_j ( D_{rj})_{r \in [t]}  :
              C_t \;|\; X_1 \dotsb X_m S_\rA
              (E_{r-1} D_{r1} D_{r2} \dotsb D_{r(j-1)})_{r \in [t]}
              R_1 \dotsb R_{j-1} )_{\rho_t} \enspace.
\end{align*}
Using the chain rule, we get
\begin{align*}
\MoveEqLeft \sum_{j \in [m]} \QIC_{\rB \rightarrow \rA} (\Pi_j, \mu_0)  \\
    & \le \quad \sum_{t \in [T]} \rI(R_1 \dotsb R_m 
                (D_{r1} D_{r2} \dotsb D_{rm})_{r \in [t]} 
                 :  C_t \;|\; X_1 \dotsb X_m S_\rA ( E_{r-1})_{r \in [t]})_{\rho_t}
            \enspace.
\end{align*}
Since each summand in the expression above is bounded by~$2 \log|C_t|$ and~$\log|C_t| \le s$, the sum is bounded by~$2sT$.
It follows that there exists an index $j^*$ such that
\begin{align*}
\QIC_{\rB \rightarrow \rA} (\Pi_{j^*}, \mu_0) \quad \leq \quad 2sT / m \enspace,
\end{align*}
as desired. As noted before, $\QIC_{\rA \rightarrow \rB} (\Pi_{j^*}, \mu_0)
\leq 2sT$. This completes the reduction.
\end{proof}

\section{Key Technical Tools}
\label{sec:tools}

In this section, we develop some tools needed to analyze the quantum
information cost of protocols.

In analyzing safe quantum protocols with classical inputs in the rest of
the paper, we deviate slightly from the notation for the registers used
in the definition of two-party protocols in Section~\ref{sec-qcc}.
We refer to the 
input registers~$A_\rin, B_\rin$ by~$X,Y$, respectively. Since we 
focus on safe protocols, the registers~$XY$ are only used as 
control registers. We express Alice's local registers after the~$i$th
message is generated as~$X A_i$, and the local registers of Bob by~$Y B_i$.
As before, the message register is not included in any of the local
registers, and is denoted by~$C_i$.

\subsection{Superposition-Average Encoding Theorem}

We first generalize the Average Encoding
Theorem~\cite{KNTZ07}, to relate the quality of approximation of 
any intermediate state in a two-party quantum communication protocol to
its information cost. This also allows us to analyze states arising from
arbitrary superpositions over inputs in such protocols.
Informally, the resulting statement is that if the incremental
information contained in the messages received by a party is
``small'', then she can approximate her part of the joint state 
 ``well'', \emph{without any assistance from the other party\/}.
The main technical ingredient of its proof is 
the Fawzi-Renner inequality~\cite{FawziR15}.

\begin{theorem}[Superposition-Average Encoding Theorem]
Given any safe quantum protocol $\Pi$ with input registers $XY$
initialized according to distribution~$\nu$, let
\begin{align*}
\ket{\rho_{i}} \quad \eqdef \quad  \sum_{x, y} \sqrt{ \nu (x, y)} \ket{xxyy}^{X R_X Y R_Y}  \ket{\rho_i^{xy}}^{A_i B_i C_i}
\end{align*}
be the state of the registers $X  Y R A_i B_i C_i$ in round $i$ 
with the register $R$ initially purifying the registers $XY$, with a
decomposition $R_X R_Y$ into coherent copies of $X$ and $Y$, respectively.
Let~$\epsilon_i \eqdef \rI (R \!:\! C_i \,|\,  Y B_i)$ for odd $i$, and
$\epsilon_i \eqdef \rI(R \!:\! C_i \,|\,  X A_i)$ for even $i$.
There exist registers $E_i$, isometries $V_i$ and states~$\ket{\theta_i}$ and~$\ket{\theta^y_i}$ with
\begin{align*}
\ket{\theta_i} \quad = \quad  \sum_{x, y} \sqrt{ \nu (x, y)} \ket{xxyy}^{X R_X Y R_Y}   \ket{\theta_i^{y}}^{B_i C_i E_i}
\end{align*}
for odd $i$ satisfying
\begin{align}
\label{eq-hb}
\fh  \!\left(\rho_i^{R Y B_i C_i}, \theta_i^{R Y B_i C_i} \right)
    \quad & \leq \quad \sum_{p \leq i,~p~\mathrm{odd} }
    \sqrt{\epsilon_p} \enspace, \qquad \textrm{and} \\
\nonumber
V_i \ket{y}^Y 
    \quad & = \quad \ket{y}^Y \otimes \ket{\theta_i^y}^{B_i C_i E_i} \enspace,
\end{align}
and states~$\ket{\sigma_i}$ and~$\ket{\sigma^x_i}$ with
\begin{align*}
\ket{\sigma_i} \quad = \quad  \sum_{x, y} \sqrt{ \nu (x, y)} \ket{xxyy}^{X R_XY R_Y}   \ket{\sigma_i^{x}}^{A_i C_i E_i}
\end{align*}
for even $i$ satisfying
\begin{align}
\label{eq-ha}
\fh  \!\left(\rho_i^{R X A_i C_i}, \sigma_i^{R X A_i C_i} \right) 
    \quad & \leq \quad \sum_{p \leq i,~p~\mathrm{even} }
    \sqrt{\epsilon_p} \enspace, \qquad \textrm{and} \\
\nonumber
V_i \ket{x}^X
    \quad & = \quad \ket{x}^X \otimes \ket{\sigma_i^x}^{A_i C_i E_i}
    \enspace.
\end{align}

\end{theorem}

\begin{proof}
The proofs for odd and even $i$'s are similar; we focus on even $i$'s.
Given $\epsilon_p \eqdef \rI(R \!:\! C_p \,|\,  X A_p)$ for even $p$, let $\R_p^{X A_p \rightarrow X A_p C_p}$ 
be the recovery map given by the Fawzi-Renner inequality, Lemma~\ref{lem:frineq}, such that
\begin{align*}
\fh  \!\left(\rho_p^{R X A_p C_p}, \R_p (\rho_p^{R X A_p}) \right) \quad \leq \quad
\sqrt{\epsilon_p} \enspace,
\end{align*}
and take $U_{\R_p}^{X A_p \rightarrow X A_p C_p F_p}$ to be an isometric extension of~$\R_p$.
Since register $R$ contains a coherent copy of~$X$, we can assume without loss of generality 
that $U_{\R_p}^{X A_p \rightarrow X A_p C_p F_p}$ uses~$X$ as a control register, i.e., acts as
\begin{align*}
U_{\R_p}^{X A_p \rightarrow X A_p C_p F_p} \ket{x}^X \ket{\phi}^{A_p G_p} \quad
= \quad \ket{x}^X \ket{\phi_x}^{A_p C_p F_p G_p} \enspace,
\end{align*}
for any state~$\ket{\phi}$ of~$A_p$ and ancillary register~$G_p$, and any input~$x$.
Consider an isometry $V_0^{X \rightarrow XT_\rA T_\rB}$ such that for all $x$,
\begin{align*}
V_0 \ket{x}^X \quad = \quad \ket{x}^X \ket{\psi}^{T_\rA T_\rB} \enspace,
\end{align*}
i.e.,~$V_0$ is an isometry that locally creates the same state
$\ket{\psi}^{T_\rA T_\rB}$, used as shared entanglement in~$\Pi$,
for any input~$x$. Let~$\rho_\nu$ denote the purified initial state of
the protocol:
\begin{align*}
\ket{\rho_\nu}
\quad \eqdef \quad \sum_{x, y} \sqrt{\nu (x, y)} \ket{xxyy}^{X R_X Y R_Y}
\enspace.
\end{align*}
We show that the isometry~$V_i$, state~$\sigma_i$, and register~$E_i$
defined as follows
\begin{align*}
V_i \quad & \eqdef \quad U_{\R_i} U_{i-1} \cdots U_{\R_4} U_{3} U_{\R_2} U_{1}
V_0 \enspace, \\
\ket{\sigma_i}
    \quad & \eqdef \quad V_i \ket{\rho_\nu} \enspace, \qquad \textrm{and} \\
E_i \quad & \eqdef \quad T_\rB \otimes C_1 \otimes F_2 \otimes C_3 \otimes F_4 \otimes \cdots \otimes C_{i-1} \otimes F_i
\end{align*}
satisfy the conditions of the theorem.
We show this by induction on $i$. 

First, note that $V_i$ is of the desired form, and uses $X$ as a control register.
For the base case, $i=2$, we start with
\begin{align*}
\ket{\rho_0}^{XYR_X R_Y T_\rA T_\rB} \quad = \quad V_0
\ket{\rho_\nu}^{XY R_X R_Y} \enspace,
\end{align*}
apply $U_1^{X T_\rA \rightarrow X A_1 C_1}$ to obtain $\ket{\rho_1}^{XY
R_X R_Y A_1 C_1 T_\rB}$,
and furthermore apply
$U_2^{Y C_1 T_\rB \rightarrow Y C_2 B_2}$ to obtain 
$\ket{\rho_2} = U_2 U_1 V_0 \ket{\rho_\nu}^{XY R_X R_Y}$.
It holds that 
\begin{align*}
\fh  \!\left(\rho_2^{R X A_2 C_2} , \R_2^{X A_2 \rightarrow X A_2 C_2}
(\rho_2^{R X A_2}) \right) & \quad
= \quad 
	\fh  \!\left(\rho_2^{R X A_2 C_2} , \R_2^{X A_1 \rightarrow X
A_2 C_2} (\rho_1^{R X A_1}) \right) \\
			& \quad \leq \quad \sqrt{\epsilon_2} \enspace,
\end{align*}
in which we used that  $\rho_2^{R X A_2} = \rho_1^{R X A_1}$ since the registers $Y C_2 B_2$ on which $U_2$ acts have 
been traced out, and $A_2 = A_1$.
Since it also holds that $\R_2  (\rho_1^{R X A_1}) = \Tr{ Y T_\rB C_1
F_2 }{ U_{\R_2} U_1 V_0 \ket{\rho_\nu}} = \sigma_2^{R X A_2 C_2}$, 
Eq.~(\ref{eq-ha}) follows.

For the induction step, we note that for even $i  > 2$,
$V_i = U_{\R_i} U_{i-1} V_{i-2}$, $E_i = F_i \otimes C_{i-1} \otimes
E_{i-2}$, and $\ket{\sigma_i} = U_{\R_i} U_{i-1} \ket{\sigma_{i-2}}$. 
Eq.~(\ref{eq-ha}) results from the following chain of inequalities, which are explained below:
\begin{align*}
\fh  \!\left(\rho_i^{R X A_i C_i} , \sigma_i^{R X A_i C_i} \right)  \quad \leq \quad &
\fh  \!\left(\rho_i^{R X A_i C_i} , \R_i^{X A_i \rightarrow X A_i C_i}
(\rho_i^{R X A_i}) \right)  \\
	& \mbox{} + \fh  \!\left( \R_i^{X A_i \rightarrow X A_i C_i}
(\rho_i^{R X A_i}), \sigma_i^{R X A_i C_i} \right)  \\
	\quad \leq \quad & \sqrt{\epsilon_i} + \fh  \!\left(\rho_i^{R X
A_i}, \Tr{C_{i-1} Y E_{i-2}}{ U_{i-1} \ket{\sigma_{i-2}}} \right) \\
	\quad \leq \quad & \sqrt{\epsilon_i} + \fh
\!\left(\rho_{i-2}^{R X A_{i-2} C_{i-2}}, \sigma_{i-2}^{R X A_{i-2}
C_{i-2}} \right) \\
	\quad \leq \quad & \sqrt{\epsilon_i} + \sum_{p \leq
i-2,~p~\mathrm{even}} \sqrt{\epsilon_p} \enspace.
\end{align*}
The first step is an application of the triangle inequality, and the
second follows by the definition of $\R_i$ 
and monotonicity of $h$ under the CPTP map $\R_i = \Tra{F_i} \circ \,
U_{\R_i}$. The third inequality holds because 
$\rho_i^{R X A_i} = \rho_{i-1}^{R X A_{i-1}} =  
\Tr{C_{i-1} Y B_{i-2}}{ U_{i-1} \ket{\rho_{i-2}}}$, the
isometry~$U_{i-1}$
does not act on registers~$E_{i-2}$ or~$YB_{i-2}$,
and by the monotonicity of $h$ under the map $\Tra{C_{i-1}} \circ \, U_{i-1}$. 
The last inequality holds by the induction hypothesis. This completes the proof of the theorem.
\end{proof}

\subsection{Quantum Cut-and-Paste Lemma}

The direct quantum analogue to the Cut-and-Paste Lemma~\cite{Bar-YossefJKS04} 
from classical communication complexity does not hold.
We can nevertheless obtain the following weaker property, linking the
states in a two-party protocol corresponding to any four possible pairs 
of inputs in a two-by-two rectangle. 
The result implies that if the states corresponding to two inputs~$x,x'$
to Alice and a fixed input~$y$ to Bob are close up to 
a local unitary operation on Alice's side, and the states for two
inputs~$y,y'$ to Bob and a fixed input~$x$ to Alice
are close up to a local unitary operation on Bob's side, then, up to
local unitary operations on Alice's and Bob's sides, the states for 
all pairs $(x^{\prime \prime}, y^{\prime \prime})$ of inputs in the 
rectangle $\{x, x^\prime \} \times \{y, y^\prime \}$ 
are close.  The lemma is a variant of the hybrid argument developed 
in Refs.~\cite{JainRS03b, JainN14}, and is proven along the same lines.
A similar, albeit slightly weaker statement may be derived
from the corresponding lemmas in these articles. For example,
Lemma~IV.10 from
Ref.~\cite{JainN14}, when adapted to the setting described above and
combined with a triangle inequality, implies bounds similar to those
in Eqs.~(\ref{eq:indodd}) and~(\ref{eq:indeven}) below. However, in
the notation of the lemma below, the bounds so derived are both larger by
the additive term~$2 h_{i-1}$.
\begin{lemma}[Quantum Cut-and-Paste]
Given any safe quantum protocol $\Pi$ with classical inputs, 
consider distinct inputs $x, x^\prime$ for Alice, and~$y, y^\prime$ for
Bob. Let $\ket{\rho_0}^{A_0 B_0}$ be 
the shared initial state of Alice and Bob for any pair of inputs
$(x^{\prime \prime}, y^{\prime \prime}) \in \{x, x^\prime \} \times \{y, y^\prime \}$. (The state~$\rho_0$ may depend on the set $\{x, x^\prime\} \times \{y, y^\prime\}$.)
Let $\ket{\rho_{i, x^{\prime \prime} y^{\prime \prime}}}^{A_i B_i C_i}$ be the state of 
registers $A_i B_i C_i$ after the~$i$th message is sent, when the input
is~$(x^{\prime \prime}, y^{\prime \prime})$.
For odd $i$, let
\begin{align*}
h_i \quad \eqdef \quad \fh  \!\left(\rho_{i, xy}^{B_i C_i}, \rho_{i,x^\prime y}^{B_i
C_i} \right)
\end{align*}
and $V_{i, x \rightarrow x^\prime}^{A_i}$ denote the unitary operation
acting on $A_i$ given by the Local Transition Lemma
(Lemma~\ref{lem:loctr})  such that 
\begin{align*}
h_i \quad = \quad \fh  \!\left(V_{i, x \rightarrow x^\prime}^{A_i} \ket{\rho_{i,
xy}} \enspace,
\ket{\rho_{i, x^\prime y}}\right) \enspace.
\end{align*}
For even $i$, let
\begin{align*}
h_i \quad \eqdef \quad \fh  \!\left(\rho_{i, xy}^{A_i C_i}, \rho_{i,x y^\prime}^{A_i
C_i} \right)
\end{align*}
and $V_{i, y \rightarrow y^\prime}^{B_i}$ denote the unitary operation acting on $B_i$ given by the Local Transition Lemma  such that 
\begin{align*}
h_i \quad = \quad \fh  \!\left(V_{i, y \rightarrow y^\prime}^{B_i} \ket{\rho_{i, xy}},
\ket{\rho_{i, x y^\prime}}\right) \enspace.
\end{align*}
Define $V_{0, y \rightarrow y^\prime}^{B_0} \eqdef \bI^{B_0}$ and $h_0 \eqdef
0$. Recall that $B_i = B_{i-1}$ for odd~$i$ and $A_i = A_{i-1}$ for even $i$.
It holds that for odd $i$,
\begin{align}
\label{eq:eqodd}
\fh  \!\left(V_{i-1, y \rightarrow y^\prime}^{B_i} \ket{\rho_{i, xy}},
\ket{\rho_{i, x y^\prime}}\right) & \quad
= \quad h_{i-1} \enspace, \\
\label{eq:indodd}
\fh  \!\left(V_{i, x \rightarrow x^\prime}^{A_i} V_{i-1, y \rightarrow
y^\prime}^{B_i} \ket{\rho_{i, xy}}, \ket{\rho_{i, x^\prime
y^\prime}}\right) & \quad
\leq \quad h_i + h_{i-1 } + 2 \sum_{j=1}^{i-2} h_j \enspace,
\end{align}
and for even $i$, 
\begin{align}
\label{eq:eqeven}
\fh  \!\left(V_{i-1, x \rightarrow x^\prime}^{A_i} \ket{\rho_{i, xy}},
\ket{\rho_{i, x^\prime y}}\right) & \quad
= \quad h_{i-1} \enspace, \\
\label{eq:indeven}
\fh  \!\left(V_{i, y \rightarrow y^\prime}^{B_i} V_{i-1, x \rightarrow
x^\prime}^{A_i} \ket{\rho_{i, xy}}, \ket{\rho_{i, x^\prime
y^\prime}}\right) & \quad
\leq \quad h_i + h_{i-1} +  2 \sum_{j=1}^{i-2} h_j \enspace.
\end{align}
\end{lemma}

\begin{proof}
We have $\ket{\rho_{0, x^{\prime \prime} y^{\prime \prime}}} =
\ket{\rho_0}$, and define $C_0$ to be a trivial register.
For odd $i$, let $U_{i, x^{\prime \prime}}$ be the protocol isometry $U_i$ conditional on the state of $X$ being $\ket{x^{\prime \prime}}$. Then we have
\begin{align*}
U_{i, x^{\prime \prime}}^{A_{i-1} C_{i-1} \rightarrow A_i C_i} \ket{\rho_{i-1, x^{\prime \prime} y^{\prime \prime}}}^{A_{i-1} B_{i-1} C_{i-1}} \quad
= \quad \ket{\rho_{i, x^{\prime \prime} y^{\prime \prime}}}^{A_{i} B_{i}
C_{i}} \enspace.
\end{align*}
It follows by the isometric invariance of $h$ and because $V_{i-1, y \rightarrow y^\prime}$ and $U_{i, x}$ act on disjoint sets of registers that
\begin{align*}
\fh  \!\left(V_{i-1, y \rightarrow y^\prime}^{B_i} \ket{\rho_{i, xy}},
\ket{\rho_{i, x y^\prime}}\right) & \quad
= \quad 
	\fh  \!\left(V_{i-1, y \rightarrow y^\prime}^{B_{i-1}} \ket{\rho_{i-1,
xy}}, \ket{\rho_{i-1, x y^\prime}}\right) \\
		& \quad = \quad h_{i-1} \enspace.
\end{align*}
Similarly, for even $i$,
\begin{align*}
U_{i, y^{\prime \prime}}^{B_{i-1} C_{i-1} \rightarrow B_i C_i} \ket{\rho_{i-1, x^{\prime \prime} y^{\prime \prime}}}^{A_{i-1} B_{i-1} C_{i-1}} \quad
= \quad \ket{\rho_{i, x^{\prime \prime} y^{\prime \prime}}}^{A_{i} B_{i} C_{i}}
\end{align*}
and
\begin{align*}
\fh  \!\left(V_{i-1, x \rightarrow x^\prime}^{A_i} \ket{\rho_{i, xy}},
\ket{\rho_{i, x^\prime y}}\right) 
		& \quad = \quad h_{i-1} \enspace.
\end{align*}
This proves Eqs.~(\ref{eq:eqodd}) and~(\ref{eq:eqeven}).

We show by induction on $i$ that Eqs.~(\ref{eq:indodd}) and~(\ref{eq:indeven}) 
hold for odd and even $i$'s, respectively. For the base case ($i=1$), we
have that
\begin{align*}
\ket{\rho_{1, x^{\prime \prime} y^{\prime \prime}}}^{A_1 B_1 C_1} 
    \quad = \quad U_{x^{\prime \prime}}^{A_0 \rightarrow A_1 C_1} \ket{\rho_0}^{A_0 B_0} \enspace,
\end{align*}
and is independent of $y^{\prime \prime}$. So
\begin{align*}
\ket{\rho_{1, x y}} \quad = \quad \ket{\rho_{1, x y^{\prime}}} \enspace,
\qquad \mathrm{and} \qquad
\ket{\rho_{1, x^\prime y}}
    \quad = \quad \ket{\rho_{1, x^\prime y^{\prime}}} \enspace,
\end{align*}
and Eq.~(\ref{eq:indodd}) follows by taking $V_{0, y \rightarrow y^\prime}^{B_0} \eqdef
\bI^{B_0}$:
\begin{align*}
\fh  \!\left(V_{1, x \rightarrow x^\prime}^{A_1} \ket{\rho_{1, xy}},
\ket{\rho_{1, x^\prime y^\prime}}\right) & \quad
= \quad h_1 \enspace.
\end{align*}

For the induction step (with~$i \ge 2$), the case of even and odd $i$'s are proven similarly; we focus on even $i$'s. 
Assuming that Eq.~(\ref{eq:indodd}) holds for $i-1$, we show Eq.~(\ref{eq:indeven}) holds for $i$ by the following chain of inequalities which are explained below:
\begin{align*}
\fh  \!\left(V_{i, y \rightarrow y^\prime}^{B_i} V_{i-1, x \rightarrow
x^\prime}^{A_i} \ket{\rho_{i, xy}}, \ket{\rho_{i, x^\prime
y^\prime}}\right)
  \quad \leq \quad & \fh \!\left(V_{i, y \rightarrow y^\prime}^{B_i} V_{i-1, x
\rightarrow x^\prime}^{A_i} \ket{\rho_{i, xy}}, V_{i-1, x \rightarrow
x^\prime}^{A_i} \ket{\rho_{i, x y^\prime}}\right) \\
& + \fh \!\left(V_{i-1, x \rightarrow x^\prime}^{A_i} \ket{\rho_{i, x
y^\prime}}, \ket{\rho_{i, x^\prime y^\prime}}\right) \\
 \quad = \quad & h_i + \fh \!\left(V_{i-1, x \rightarrow x^\prime}^{A_{i-1}}
\ket{\rho_{i-1, x y^\prime}}, \ket{\rho_{i-1, x^\prime y^\prime}}\right) \\
\quad \leq \quad & h_i + \fh \!\left(V_{i-1, x \rightarrow x^\prime}^{A_{i-1}}
\ket{\rho_{i-1, x y^\prime}}, V_{i-1, x \rightarrow x^\prime}^{A_{i-1}}
V_{i-2, y \rightarrow y^\prime}^{B_{i-1}} \ket{\rho_{i-1, x y}}\right) \\
& + \fh \!\left(V_{i-1, x \rightarrow x^\prime}^{A_{i-1}} V_{i-2, y \rightarrow
y^\prime}^{B_{i-1}} \ket{\rho_{i-1, x y}}, \ket{\rho_{i-1, x^\prime
y^\prime}}\right) \\
  \quad \leq \quad & h_i + h_{i-2} + h_{i-1} + h_{i-2} +  2
\sum_{j=1}^{i-3} h_j \enspace.
\end{align*}
The first step is by the triangle inequality. In the second step, we
used the unitary invariance of $\fh$ along with the definition of $h_i$
for the first term, and with the property that $V_{i-1}^{A_i}$ 
and $U_{i, y^{\prime }}^{B_{i-1} C_{i-1} \rightarrow B_i C_i}$ act on disjoint sets of registers 
for the second term. The next step follows by the triangle inequality. The last step is
by Eq.~(\ref{eq:eqodd}) for the second term, and by the induction hypothesis (Eq.~(\ref{eq:indodd})) for the third term.
\end{proof}

\section{QIC Lower Bound for Augmented Index}
\label{sec:QIC-lb-AI}

In this section, we establish a lower bound for the quantum information
cost of protocols for Augmented Index.

\subsection{Relating Alice's states to $\QIC_{\rB \rightarrow A}$}
\label{sec:genalice}

We study the quantum information cost of protocols for Augmented Index
on input distribution $\mu_0$ (the uniform distribution
over~$f_n^{-1}(0)$),
and relate it to the distance between the states on two different inputs.
We first focus on the quantum information cost from Bob to Alice,
arising from the messages with even~$i$'s. We show that if this cost is
low, then Alice's reduced states on different inputs for Bob are close
to each other. (This high level intuition is the same as that described
in Ref.~\cite{JainN14}.)

We state and prove our results for inputs with even length $n$; a similar result can be shown
for odd $n$ by suitably adapting the proof.

We consider the following purification of the input registers, corresponding to a particular
preparation method for the $K$ register, and to a preparation of the $X$
register also depending on the preparation of register~$K$. Recall that
the content $k$ of register $K$ is uniformly distributed in $[n]$. 
The following registers 
are each initialized to uniform superpositions over the domain
indicated: $R_S^1$ over~$\set{0,1}$ (with a coherent copy in $R_S^2$),
register~$R_J^1$ over indices $j \in [n/2]$ (with a coherent copy in $R_J^2$),
register $R_L^1$ over $\ell \in [n/2+1, n]$ (with a coherent copy in $R_L^2$).
Register~$R_K$ holds a coherent copy of register~$K$, whose content $k$
is set to the value~$j$ in~$R_J^1$ when $R_S^1$ 
is~$0$, and  to~$\ell$ when~$R_S^1$ is~$1$.  Depending on the
value~$\ell$ of~$R_L^1$, 
the following registers are initialized to uniform superpositions to prepare the $X$ register, 
itself uniform over $\{0, 1 \}^n$: register 
$R_Z^1$ over $z \in \{0, 1 \}^\ell$, and register $R_W^1$ over~$w \in \{0, 1 \}^{n - \ell}$. The register $X$ 
is set to $x= zw$, so together $R_Z^1 R_W^1$ hold a coherent copy of $X$, 
and a second coherent copy is held in  $R_Z^2 R_W^2$. 
If $\ell$ is clear from the context, we sometimes use the notation $Z$ and $W$ to 
refer to the parts of the $X$ register holding $z$ and $w$, respectively.
Depending on the value~$j$ of~$R_J^1$, we also refer to a further 
decomposition $z = z^{\prime} z^{\prime \prime}$ with $z^\prime \in \{0, 1\}^j$ and 
$z^{\prime \prime} \in \{0, 1 \}^{\ell - j}$. We denote by $X_{1K}$ the register 
held by Bob, containing the first $k-1$ bits of $x$ and the verification bit~$b$. The bit~$b$ is always equal to $x_k$ under $\mu_0$; $X_{1K}$ thus contains the first $k$ bits of $x$ in this case. The register~$X_{1K}$ is set to 
$z^\prime$ when $R_S^1$ is $0$, to $z$ when $R_S^1$ is 1, 
and register $R_{X_{1K}}$ holds a coherent copy of~$X_{1K}$.

In summary, the resulting input state $\rho_{\mu_0}^{XKX_{1K}}$ 
distributed according to $\mu_0$ is purified by register $R$, which 
decomposes as
\[
R \quad \eqdef \quad
    R_J^1 R_J^2 R_L^1 R_L^2 R_Z^1 R_Z^2 R_W^1 R_W^2 R_S^1 R_S^2 R_K R_{X_{1K}}
    \enspace.
\]
Using the normalization factor $c \eqdef
1 / \sqrt{ (n/2) \cdot (n/2)  \cdot 2^{\ell} \cdot 2^{n - \ell} \cdot
2}$, the purified state is:
\begin{align}
\label{eq:purif_alice}
\ket{\rho_0}^{R X K X_{1K}} \quad \eqdef \quad c \sum_{j, \ell, z, w} \ket{jj \ell \ell zz ww} \left(
\ket{00} \ket{j z^\prime} \ket{zw}^{X} \ket{j z^\prime}^{K X_{1K}}  +
\ket{11} \ket{\ell z} \ket{zw}^X \ket{\ell z}^{K X_{1K}} \right)
\enspace.
\end{align}

Starting with the above purification and using shared entanglement
$\ket{\psi}^{T_\rA T_\rB}$ in the initial state, 
the state $\rho_i$ over the registers~$R X K X_{1K} A_i B_i C_i$ 
after round $i$ in the protocol is
\begin{align}
\label{eq:purif_alice_rdi}
\ket{\rho_i} \quad \eqdef \quad c \sum_{j, \ell, z, w} \ket{jj \ell \ell zz ww} \left(
\ket{00}  \ket{j z^\prime} \ket{zw} \ket{j z^\prime} \ket{\rho_i^{zw,
(j, z^\prime)}}  + \ket{11}  \ket{\ell z} \ket{zw} \ket{\ell z}
\ket{\rho_i^{zw, ( \ell, z)}} \right) \enspace,
\end{align}
where~$\ket{\rho_i^{x, (k, x[1, k])}}$ denotes
the pure state in registers $A_i B_i C_i$ conditional on joint input~$\left(
x, (k, x[1, k])\right)$.

Define~$R_\rA \eqdef  R_J^1 R_L^1 R_S^1 R_K  R_W^1 R_W^2$. All of $R_\rA$'s
sub-registers except $R_W^1 R_W^2$ are classical in $\rho_i^{R_A X A_i
C_i}$, since one
of their coherent copies is traced out from the global purification register $R$. The $Z$ part of 
the $X$ register is also classical.
We can write the reduced state of~$\rho_i$ on registers $R_A X A_i C_i$ as
\begin{align*}
\rho_i^{R_A X A_i C_i} \quad = \quad& c^\prime \sum_{j, \ell, z} \kb{j \ell }{j
\ell } \otimes \left( \kb{0j}{0j} \otimes \kb{z}{z}^Z \otimes \rho_{i,
\ell  z  j z^\prime}  + \kb{1 \ell }{1 \ell } \otimes \kb{z}{z}^Z
\otimes \rho_{i,  \ell z \ell z} \right) \enspace,
\end{align*}
in which we used normalization $c^\prime \eqdef 1 / ( (n/2) \cdot (n/2) \cdot 2^\ell \cdot 2)$ and the shorthands 
\begin{align}
\label{eq:rhoilzk}
\rho_{i,  \ell z  k x[1, k]} \quad \eqdef \quad& \Tr{B_i}{\kb{\rho_i^{ \ell z
k x[1, k]}}{\rho_i^{\ell z k x[1, k]}}} \enspace, \qquad \textrm{where} \\
\label{eq:purerhoilzk}
\ket{\rho_i^{ \ell z  k x[1, k]}} \quad \eqdef \quad& 1/ \sqrt{2^{n - \ell}}
\sum_w \ket{www}^{R_W^1 R_W^2 W} \ket{\rho_i^{zw, (k, x[1, k])}}^{A_i
B_i C_i} \enspace.
\end{align}
The indices $\ell z k x[1, k]$ have the following meaning: $\ell$ and $z$ indicate that 
Alice's input register~$X$ is in superposition after the length $\ell$
prefix $z$ (with~$z = x[1, \ell]$); $k$ and $x[1, k]$ tell us
the index $k$ in Bob's input, the prefix $x[1, k-1]$ 
of $x$ given as input to Bob, and Bob's verification bit $b$
(which is equal to $x_k$ under $\mu_0$), respectively.
Using this notation along with the superposition-average encoding theorem, we show the following result.

\begin{lemma}
\label{lem:qicvsalice}
Given any even $n \ge 2$, let $J$ and $L$ be random variables uniformly
distributed in $[n/2]$ and $[n] \setminus [n/2]$, respectively. Conditional on some 
value $\ell$ for $L$, let $Z$ be a random variable chosen uniformly at random in $\{0, 1 \}^\ell$. The following 
then holds for any $M$-message safe quantum protocol $\Pi$ for Augmented
Index~$f_n \,$, for any even $i \leq M$:
\begin{align*}
 \QIC_{\rB \rightarrow A} (\Pi, \mu_0) \quad \geq \quad \frac{1}{2 M} \; \mathbb{E}_{j
\ell z\sim JLZ} \left[ \fh^2 \!\left(\rho_{i, \ell z j z^\prime}^{  R_W^1 R_W^2
W A_i C_i}, \rho_{i, \ell z \ell z}^{  R_W^1 R_W^2 W A_i C_i} \right)\right]
\enspace.
\end{align*}
\end{lemma}

\begin{proof}
Considering the same purification of the input state as in Eq.~(\ref{eq:purif_alice}), we get the following states from the superposition-average 
encoding theorem 
\begin{align*}
\ket{\sigma_i} \quad \eqdef \quad & c \sum_{j, \ell, z, w} \ket{jj \ell \ell zz ww}
\left( \ket{00}  \ket{j z^\prime} \ket{zw} \ket{j z^\prime}
\ket{\sigma_i^{zw}}^{A_i C_i E_i}  + \ket{11}  \ket{\ell z} \ket{zw}
\ket{\ell z} \ket{\sigma_i^{zw}}^{A_i C_i E_i} \right) \enspace, 
\end{align*}
satisfying
\begin{align*}
\fh (\rho_i^{R XA_iC_i}, \sigma_i^{R XA_iC_i}) \quad \leq \quad & \sum_{p \leq
i,~p~\mathrm{even}} \sqrt{\rI (C_p\!:\!R\,|\, X A_p)} \enspace.
\end{align*}
The reduced state of~$\sigma_i$ on registers $R_\rA X A_i C_i$ is
\begin{align*}
\sigma_i^{R_\rA X A_i C_i} \quad = \quad& c^\prime \sum_{j, \ell, z} \kb{j \ell }{j
\ell } \otimes \left( \kb{0j}{0j} \otimes \kb{z}{z}^Z \otimes \sigma_{i,
\ell  z }  + \kb{1 \ell }{1 \ell } \otimes \kb{z}{z}^Z \otimes
\sigma_{i,  \ell z } \right) \enspace,
\end{align*}
in which we use the shorthands 
\begin{align*}
\sigma_{i, \ell z} \quad \eqdef \quad& \Tr{E_i}{\kb{\sigma_i^{ \ell
z}}{\sigma_i^{\ell z}}} \enspace, \quad \textrm{where} \\
\ket{\sigma_i^{ \ell z}} \quad \eqdef \quad& \frac{1}{\sqrt{2^{n - \ell}}}
\; \sum_w
\ket{www}^{R_W^1 R_W^2 W} \ket{\sigma_i^{  zw}}^{A_i C_i E_i} \enspace.
\end{align*}
The lemma then follows from the next chain of inequalities, as explained below:
\begin{align*}
\frac{i}{2} \sum_{p \leq i,~p~\mathrm{even}} \rI (C_p\!:\!R\,|\, X A_p)
\quad \geq \quad & \left(\sum_{p \leq i,~p~\mathrm{even}} \sqrt{\rI
(C_p\!:\!R\,|\, X A_p)}\right)^2 \\
			\quad \geq \quad & \fh^2 \!\left(\rho_i^{ R_\rA Z
W A_i C_i}, \sigma_i^{R_\rA Z W A_i C_i} \right) \\
			\quad = \quad &  \frac{1}{2} \; \mathbb{E}_{j \ell z \sim JLZ}
\left[ \fh^2 \!\left(\rho_{i, \ell z j z^\prime}^{  R_W^1 R_W^2 W A_i
C_i}, \sigma_{i,  \ell z}^{  R_W^1 R_W^2 W A_i C_i} \right) \right]  \\
			& + \frac{1}{2} \; \mathbb{E}_{j \ell z \sim JLZ}
\left[ \fh^2 \!\left(\rho_{i, \ell z \ell z}^{  R_W^1 R_W^2 W A_i C_i},
\sigma_{i,  \ell z}^{  R_W^1 R_W^2 W A_i C_i} \right) \right] \\
			\quad \geq \quad & \frac{1}{4} \;  \mathbb{E}_{j \ell z \sim
JLZ} \left[ \fh^2 \!\left(\rho_{i, \ell z j z^\prime}^{  R_W^1 R_W^2 W A_i
C_i}, \rho_{i, \ell z \ell z}^{  R_W^1 R_W^2 W A_i C_i} \right) \right]
\enspace.
\end{align*}
The first inequality is by the concavity of the square root function and
the Jensen inequality, and the second by the superposition-average encoding 
theorem along with the monotonicity of $\fh$ under tracing out a part of register~$R$.
The equality is by the joint linearity of $\fh^2$ (cf. Eq.~\eqref{eq-bures-joint-linearity}), by expanding the expectation over $R_S^1$ and by fixing $k$ 
accordingly. The last inequality is by the weak triangle inequality
of $\fh^2$ (cf. Eq.~\eqref{eq-weak-triangle}).
\end{proof}

\subsection{Relating Bob's states to $\QIC_{\rA \rightarrow \rB}$}

We continue with the notation from the previous section, and now focus
on the quantum information cost from Alice to Bob, arising from messages with odd $i$'s. 
We go via an alternative notion of information cost used by
Jain and Nayak~\cite{JainN14}, and studied further by Lauri{\`e}re and
Touchette~\cite{LT17-information-flow,LT17-information-flow-arxiv}.
This notion is similar to the internal information cost
of classical protocols (see, e.g., Refs.~\cite{Bar-YossefJKS04,BarakBCR13}),
and is called the Holevo information cost in Ref.~\cite{LT17-information-flow-arxiv}.
\begin{definition}
Given a safe quantum protocol $\Pi$ with classical inputs, and
distribution $\nu$ over inputs, the \emph{Holevo information cost\/}  
(of the messages) from Alice to Bob in round $i$ is defined as
\begin{align*}
\widetilde{\QIC}_{\rA \rightarrow \rB}^i (\Pi, \nu) \quad = \quad & \rI
(X \!:\! B_i C_i \,|\, 
Y ) \enspace,
\end{align*}
and the cumulative \emph{Holevo information cost\/} from Alice to Bob is defined as
\begin{align*}
\widetilde{\QIC}_{\rA \rightarrow \rB} (\Pi, \nu) \quad = \quad &
\sum_{i~\mathrm{odd}} \widetilde{\QIC}_{\rA \rightarrow \rB}^i (\Pi,
\nu) \enspace.
\end{align*}
\end{definition}

Given a bit string $z$ of length at least~$\ell \geq 1$, let
$z^{(\ell)}$ denote the string obtained by flipping the~$\ell$th bit of~$z$. The following result can be inferred from the proof of Lemma~4.9 in
Ref.~\cite{JainN14}.
\begin{lemma}
\label{lem:qicvsbob}
Given any even $n \geq 2$, let $J$ and $L$ be random variables uniformly
distributed in $[n/2]$ and $[n] \setminus  [n/2]$, respectively. Conditional on some 
value $\ell$ for $L$, let $Z$ be a random variable chosen uniformly at random in $\{0, 1 \}^\ell$. The following 
holds for any $M$-message safe quantum protocol $\Pi$ for the Augmented
Index function~$f_n \,$, for any odd $i \leq M$:
\begin{align*}
\frac{1}{n} \; \widetilde{\QIC}_{\rA \rightarrow \rB}^i (\Pi, \mu_0) \quad
\geq \quad
\frac{1}{16 } \; \mathbb{E}_{j \ell z\sim JLZ} \left[ \fh^2 \!\left(\rho_{i, \ell
z j z^\prime}^{  B_i C_i}~,~\rho_{i, \ell z^{(\ell)} j z^\prime}^{  B_i
C_i} \right)\right] \enspace,
\end{align*}
with $\rho_{i, \ell z j z^\prime}$ defined by Eqs.~(\ref{eq:rhoilzk}) and~(\ref{eq:purerhoilzk}).
\end{lemma}
For completeness, we provide a proof of this lemma 
in Appendix~\ref{app:jnbobqic} using our notation.

Lauri{\`e}re and Touchette~\cite{LT17-information-flow,LT17-information-flow-arxiv} prove that Holevo
information cost is a lower bound on quantum information cost. 
\begin{lemma}
\label{lem:qictilvsqicatob}
Given any $M$-message quantum protocol $\Pi$ and any input distribution $\nu$, the following holds for any odd $i \leq M$:
\begin{align*}
\widetilde{\QIC}_{\rA \rightarrow \rB}^i (\Pi, \nu) \quad \leq \quad \QIC_{\rA
\rightarrow \rB} (\Pi, \nu) \enspace.
\end{align*}
\end{lemma}
This may be derived from the Information Flow Lemma
(Lemma~\ref{thm-inf-flow}) 
by initializing the purification register~$R$ so 
that $R_a^\rB$ is a coherent copy of $X$ and $R_b^\rB$ is a coherent copy 
of $Y$, and~$R_c^\rB$ is a coherent copy of both~$X,Y$.

\subsection{Lower bound on QIC}
\label{sec:lbqic}

We are now ready to prove a slightly weaker variant of our main lower
bound (Theorem~\ref{th:formainlbaugind}) on the quantum information cost of Augmented Index.
\begin{theorem}
Given any even $n$, the following 
holds for any $M$-message safe quantum protocol $\Pi$ computing the
Augmented Index function $f_n$ with error at most $\epsilon$
on any input:
\begin{align*}
\frac{1}{4} (1 - 2 \epsilon) \quad \leq \quad &  \left( \frac{2 (M+1)^2}{n} \cdot \QIC_{\rA \rightarrow
\rB} (\Pi, \mu_0) \right)^{1/2} +  \left(  \frac{M^3}{4} \cdot \QIC_{\rB
\rightarrow \rA} (\Pi, \mu_0) \right)^{1/2} \enspace.
\end{align*}
\end{theorem}
The stronger version is proven similarly in Section~\ref{sec:altlbqic} using a strengthening of Lemma~\ref{lem:qicvsalice}.
Our argument follows ideas similar to those in Ref.~\cite{JainN14}.
Using the notation from the two previous sections, we start by fixing
values of $j, \ell, z$ in their respective domains,
 and defining $\hat{A}_i \eqdef R_W^1 R_W^2 W A_i$. Define, for odd $i$, 
\begin{align*}
h_i (j, \ell, z) & \quad \eqdef \quad \fh \!\left(\rho_{i, \ell z j z^\prime}^{  B_i C_i}~,~\rho_{i, \ell z^{(\ell)} j z^\prime}^{  B_i C_i}\right) \\
		& \quad = \quad \fh \!\left(V_{i, z \rightarrow z^{(\ell)}}^{\hat{A}_i}
\ket{\rho_{i}^{ \ell z j z^\prime}}^{ \hat{A}_i B_i C_i} \enspace,
\ket{\rho_{i}^{ \ell z^{(\ell)} j z^\prime}}^{ \hat{A}_i B_i C_i}\right)
\enspace,
\end{align*}
 and  for even $i$,
\begin{align*}
h_i (j, \ell, z) & \quad \eqdef \quad \fh \!\left(\rho_{i, \ell z j z^\prime}^{  \hat{A}_i C_i}, \rho_{i, \ell z \ell z}^{  \hat{A}_i C_i}\right) \\
		& \quad = \quad \fh \!\left(V_{i, (j, z^\prime) \rightarrow (\ell, z)}^{B_i}
\ket{\rho_{i}^{ \ell z j z^\prime}}^{  \hat{A}_i B_i C_i} \enspace,
\ket{\rho_{i}^{ \ell z \ell z}}^{  \hat{A}_i B_i C_i}\right) \enspace,
\end{align*}
where the unitary operations~$V_{i, z \rightarrow
z^{(\ell)}}^{\hat{A}_i}$ and~$V_{i, (j, z^\prime) \rightarrow (\ell,
z)}^{B_i}$ are given by the Local Transition Lemma.
We define the following states, analogous to the states 
consistent with $\mu_0$ in Eqs.~(\ref{eq:rhoilzk})
and~(\ref{eq:purerhoilzk}):
\begin{align*}
\rho_{i,  \ell z^{(\ell)}  \ell z} \quad \eqdef \quad& \Tr{B_i}{\kb{\rho_i^{
\ell z^{(\ell)} \ell z}}{\rho_i^{\ell z^{(\ell)} \ell z}}} \enspace, \\
\ket{\rho_i^{ \ell z^{(\ell)} \ell z}} \quad \eqdef \quad& \frac{1}{ \sqrt{2^{n -
\ell}}} \sum_w \ket{www}^{R_W^1 R_W^2 W} \ket{\rho_i^{z^{(\ell)} w,
(\ell, z)}}^{A_i B_i C_i} \enspace.
\end{align*}

Given protocol $\Pi$, we define a protocol $\Pi (j, \ell, z)$ that 
behaves as $\Pi$ but starts with
shared entanglement $\ket{\psi}^{T_\rA T_\rB} \otimes \ket{\psi_W}$, 
with $\ket{\psi_W} \eqdef \sqrt{1/2^{n- \ell}} \sum_w \ket{www}^{R_W^1 R_W^2
W}$ given to Alice. 
We consider runs of $\Pi (j, \ell, z)$ on four pairs of inputs: 
Alice gets inputs $u = z$ or $u^\prime = z^{(\ell)}$, and Bob gets
inputs $ y =(j, z^\prime)$ or $y^\prime = (\ell, z)$.
On these inputs~$u,u'$ of length $\ell$ for Alice, $\Pi (j, \ell, z)$ uses the content $w$ of register $W$ to obtain 
an input of length $n$ for Alice in order to run $\Pi$.
Note that regardless of $w$, the only input pair for $\Pi (j, \ell, z)$ for which Augmented Index evaluates to $1$ is $(u^\prime, y^\prime) = (z^{(\ell)}, (\ell, z))$.

If $M$ is even, denote by $\rho_{M, \ell z^{(\ell)} w k x[1, k]}^{    A_M C_M}$ the reduced state of $\rho_{M, \ell z^{(\ell)} k x[1, k]}^{    W A_M C_M}$ for a particular content $w$ of register~$W$. 
The function~$f_n$ has different values on inputs $(z^{(\ell)}w, (j,
z^\prime))$ and $(z^{(\ell)}w, (\ell, z))$.
Since the protocol $\Pi$ has error at most $\epsilon$ on any input, 
we can distinguish between these two values with probability at least $1
- \epsilon$ for any~$w$, by applying $U_{M+1, z^{(\ell)} w}$ to the
  corresponding states 
$\rho_{M, \ell z^{(\ell)} w j z^\prime}^{    A_M C_M}$ and $\rho_{M,
\ell z^{(\ell)} w \ell z}^{   A_M C_M}$.
By the relationship of trace distance with distinguishability of
states, and its monotonicity under quantum operations, we get that
\begin{align*}
\left\|\rho_{M, \ell z^{(\ell)}  j z^\prime}^{  W  A_M C_M} - \rho_{M,
\ell z^{(\ell)}  \ell z}^{  W A_M C_M} \right\|_1 
	& \quad = \quad  \frac{1}{2^{n - \ell}} \sum_w  \left\| \rho_{M,
\ell z^{(\ell)} w j z^\prime}^{    A_M C_M} - \rho_{M, \ell z^{(\ell)} w
\ell z}^{   A_M C_M} \right\|_1 \\
	& \quad \geq \quad 2 - 4 \epsilon \enspace.
\end{align*}
Here we also used the joint linearity of the trace distance (cf. Eq.~\eqref{eq-tr-joint-linearity}) to expand over
the values~$W$ takes.
We now link the quantities~$h_i (j, \ell, z)$ to the above inequality, as explained below:
\begin{align*}
\MoveEqLeft \frac{1}{\sqrt{2}}(1 - 2 \epsilon) \\
    & \leq \quad \frac{1}{2 \sqrt{2}} \left\|\rho_{M, \ell
z^{(\ell)}  j z^\prime}^{  W  A_M C_M} - \rho_{M, \ell z^{(\ell)}  \ell
z}^{  W A_M C_M} \right\|_1  \\
    & \leq \quad \fh  \!\left(\rho_{M, \ell z^{(\ell)}  j z^\prime}^{  W  A_M C_M}~,~\rho_{M, \ell z^{(\ell)}  \ell z}^{  W A_M C_M}\right) \\
    & \leq \quad \fh  \!\left(V_{M, (j, z^\prime) \rightarrow (\ell, z)}^{B_M} \ket{\rho_{M, \ell z^{(\ell)}  j z^\prime}}^{ \hat{A}_M B_M C_M}, \ket{\rho_{M, \ell z^{(\ell)}  \ell z}}^{ \hat{A}_M B_M C_M}\right) \\
    & \leq \quad \fh  \!\left(V_{M, (j, z^\prime) \rightarrow (\ell, z)}^{B_M} \ket{\rho_{M, \ell z^{(\ell)}  j z^\prime}}^{ \hat{A}_M B_M C_M}, V_{M, (j, z^\prime) \rightarrow (\ell, z)}^{B_M} V_{M - 1, z \rightarrow z^{(\ell)}}^{\hat{A}_M} \ket{\rho_{M, \ell z j z^\prime}}^{ \hat{A}_M B_M C_M}\right) \\
    & \qquad \mbox{} + \fh  \!\left(V_{M, (j, z^\prime) \rightarrow (\ell, z)}^{B_M} V_{M - 1, z \rightarrow z^{(\ell)}}^{\hat{A}_M} \ket{\rho_{M, \ell z j z^\prime}}^{ \hat{A}_M B_M C_M}, \ket{\rho_{M, \ell z^{(\ell)}  \ell z}}^{ \hat{A}_M B_M C_M}\right) \\
    & \leq \quad h_{M-1} (j, \ell, z) + h_{M} (j, \ell, z) + h_{M-1} (j, \ell, z) + 2 \sum_{i=1}^{M-2} h_{i} (j, \ell, z) \\
    & \leq \quad 2 \sum_{i=1}^M h_{i} (j, \ell, z) \enspace.
\end{align*}
The second inequality follows from Eq.~(\ref{eq:fvdg}),
and the third from monotonicity of $\fh$ under partial trace.
The fourth inequality follows by the triangle inequality, and the fifth by the Quantum Cut-and-Paste Lemma using $\hat{A}_i = R_W^1 R_W^2 W A_i$ as Alice's local register in round $i$.

In order to relate this to quantum information cost, we use Lemma~\ref{lem:qicvsalice} together with the concavity of the square root function and Jensen's inequality to obtain, for any even $i$,
\begin{align}
\label{eq:qicbavshi}
\sqrt{2 M \cdot \QIC_{\rB \rightarrow \rA} (\Pi, \mu_0)} \quad \geq
\quad \mathbb{E}_{j \ell z\sim JLZ} [h_i (j, \ell, z)] \enspace.
\end{align}
Similarly, using Lemma~\ref{lem:qicvsbob} for any odd $i$,
\begin{align*}
\sqrt{\frac{16 }{n} \cdot \widetilde{\QIC}_{\rA \rightarrow \rB}^i (\Pi, \mu_0)} \quad
\geq \quad \mathbb{E}_{j \ell z\sim JLZ} [h_i (j, \ell, z)] \enspace.
\end{align*}
Combining the above and Lemma~\ref{lem:qictilvsqicatob}, we get
\begin{align}
\label{eq:lbqicfirst}
\frac{1}{4} (1 - 2 \epsilon) \quad \leq \quad & \sum_{i \textrm{ odd}} \left( \frac{8 }{n}
\cdot \widetilde{\QIC}_{\rA \rightarrow \rB}^i (\Pi, \mu_0)
\right)^{1/2} + \sum_{i \textrm{ even}}
\left(  M \cdot \QIC_{\rB \rightarrow \rA} (\Pi, \mu_0) \right)^{1/2} \\
		\quad \leq \quad & \left( \frac{4 (M+1) }{n} 
\cdot \sum_{i \textrm{ odd}} \widetilde{\QIC}_{\rA \rightarrow \rB}^i (\Pi, \mu_0) \right)^{1/2} +
\left( \frac{ M^3}{4} \cdot \QIC_{\rB \rightarrow \rA} (\Pi, \mu_0) \right)^{1/2} \\
\label{eq:lbqiclast}
		\quad \leq \quad &  \left( \frac{2 (M+1)^2}{n} \cdot \QIC_{\rA \rightarrow
\rB} (\Pi, \mu_0) \right)^{1/2} +  \left(  \frac{M^3}{4} \cdot \QIC_{\rB
\rightarrow \rA} (\Pi, \mu_0) \right)^{1/2} \enspace,
\end{align}
which completes the proof in the case that $M$ is even.

The proof for odd $M$ is similar, and follows by comparing states
$\rho_{M, \ell z  \ell z}^{   B_M C_M}$ and $\rho_{M, \ell z^{(\ell)}
\ell z}^{   B_M C_M}$. The different outputs can then be generated by applying $U_{M+1, (\ell, z)}$ to these states.

\section{A Stronger QIC Trade-off for Augmented Index}
\label{sec:altlbqic}

In this section, we consider a different notion of quantum information cost, more specialized to the Augmented Index function, 
for which we obtain better dependence 
on $M$ for the information lower bound, from $M^3$ to $M$.
We also show that this notion is at least $1/M$ times $\QIC_{\rB
\rightarrow \rA}$, and thus we get an overall improvement by a 
factor of $M$ for the $M$-pass streaming lower bound.
The following is a precise statement of Theorem~\ref{th:infmainlbaugind}.
\begin{theorem}
\label{th:formainlbaugind}
Given any even $n$, the following 
holds for any $M$-message quantum protocol $\Pi$ computing the Augmented Index function $f_n$ with error $\epsilon$ on any input:
\begin{align*}
\frac{1}{4} (1 - 2 \epsilon) \quad \leq \quad &  \left( \frac{2 (M+1)^2}{n} \cdot \QIC_{\rA \rightarrow
\rB} (\Pi, \mu_0) \right)^{1/2} +  \left(  \frac{M^2}{2} \cdot \QIC_{\rB
\rightarrow \rA} (\Pi, \mu_0) \right)^{1/2} \enspace.
\end{align*}
\end{theorem}

Our lower bound on quantum streaming algorithms for $\dyck(2)$,
Theorem~\ref{th:infmainlbstream}, follows by combining this with
Lemmas~\ref{lem:reddycktoasc} and~\ref{lem:redaugind}, and taking $m=n$ so that $N \in \Theta (n^2)$.

We consider the same purification of the input registers as in Section~\ref{sec:genalice}, and the following alternative notion 
of quantum information cost.
\begin{definition}
Given a safe quantum protocol $\Pi$ for Augmented Index, the
\emph{superposed-Holevo information cost}  (of the messages)
from Bob to Alice in round $i$ is defined as
\begin{align*}
\widetilde{\QIC}_{\rB \rightarrow \rA}^i (\Pi, \mu_0) \quad \eqdef \quad &  \rI (R_K R_J^1
R_S^1 \!:\! R_W^1 R_W^2 W A_i C_i \,|\,  R_L^1 Z)_{\rho_i} \enspace,
\end{align*}
with $\rho_i$ as defined in Eq.~(\ref{eq:purif_alice_rdi}), and the
cumulative \emph{superposed-Holevo information cost} (of the messages) from Bob to Alice is defined as
\begin{align*}
\widetilde{\QIC}_{\rB \rightarrow \rA} (\Pi, \mu_0) \quad \eqdef \quad &
\sum_{i~\mathrm{even}} \widetilde{\QIC}_{\rB \rightarrow \rA}^i (\Pi,
\mu_0) \enspace.
\end{align*}

\end{definition}

We first show the following.
\begin{lemma}
\label{lem:impqicvsalice}
Given any even $n \ge 2$, let $J$ and $L$ be random variables uniformly
distributed in $[n/2]$ and $[n] \setminus [n/2]$, respectively. Conditional on some 
value $\ell$ for $L$, let $Z$ be a random variable chosen uniformly at
random from $\{0, 1 \}^\ell$. The following 
then holds for any $M$-message safe quantum protocol $\Pi$ for the
Augmented Index function~$f_n \,$, for even $i \leq M$:
\begin{align*}
\widetilde{\QIC}_{\rB \rightarrow \rA}^i (\Pi, \mu_0) \quad \geq \quad
\frac{1}{4 } \;
\mathbb{E}_{j \ell z\sim JLZ} \left[ \fh^2 \!\left(\rho_{i, \ell z j
z^\prime}^{  R_W^1 R_W^2 W A_i C_i}, \rho_{i, \ell z \ell z}^{  R_W^1
R_W^2 W A_i C_i} \right)\right] \enspace,
\end{align*}
with $\rho_{i, \ell z k x[1, k]}$ defined by Eqs.~(\ref{eq:rhoilzk})
and~(\ref{eq:purerhoilzk}).
\end{lemma}

\begin{proof}
The Average Encoding Theorem along with monotonicity of conditional
mutual information gives us the desired bound, with $\rho_{i, \ell z }$ the state $\rho_{i, \ell z k x[1, k]} $ 
in registers $R_W^1 R_W^2 W A_i C_i$ averaged over registers $R_K R_J^1 R_S^1$:
\begin{align*}
\MoveEqLeft
\rI (R_K R_J^1 R_S^1 \!:\! R_W^1 R_W^2 W A_i C_i \,|\,  R_L^1 Z) \\
    \geq \quad & \frac{1}{2} \; \mathbb{E}_{j \ell z \sim JLZ}
         \left[ \fh^2 \!\left(
         \rho_{i,  \ell z j z^\prime }^{R_W^1 R_W^2 W A_i C_i},
         \rho_{i, \ell z }^{R_W^1 R_W^2 W A_i C_i}
         \right) \right] \\
    &    \mbox{ } + \frac{1}{2} \; \mathbb{E}_{ j \ell z \sim JLZ}
         \left[ \fh^2 \!\left(
         \rho_{i,  \ell z  \ell z }^{R_W^1 R_W^2 W A_i C_i},
         \rho_{i, \ell z }^{R_W^1 R_W^2 W A_i C_i}
         \right) \right] \\
    \geq \quad & \frac{1}{4} \; \mathbb{E}_{j \ell z \sim JLZ }
         \left[ \fh^2 \!\left(
         \rho_{i,  \ell z j z^\prime }^{R_W^1 R_W^2 W A_i C_i},
         \rho_{i,  \ell z  \ell z}^{R_W^1 R_W^2 W A_i C_i}
         \right) \right] \enspace.
\end{align*}
\end{proof}

We now show that this notion of information cost  is a lower bound on
$\QIC_{\rB \rightarrow \rA} (\Pi, \mu_0)$.
\begin{lemma}
Given any $M$-message safe quantum protocol $\Pi$ for Augmented Index and any even $i \leq M$, the following holds:
\begin{align*}
\widetilde{\QIC}_{\rB \rightarrow \rA}^i (\Pi, \mu_0) \quad \leq \quad \QIC_{\rB
\rightarrow \rA} (\Pi, \mu_0) \enspace.
\end{align*}
\end{lemma}

\begin{proof}
The lemma is implied by the following chain of inequalities, which are
explained below.
\begin{align*}
\MoveEqLeft
\rI (R_K R_J^1 R_S^1 \!:\! R_W^1 R_W^2 W A_i C_i \,|\,  R_L^1 Z)_{\rho_i} \\
    & = \quad \rI(R_K R_J^1 R_S^1 \!:\! R_W^1 R_W^2 \,|\,  R_L^1 Z)_{\rho_i} 
        + \rI (R_K R_J^1 R_S^1 \!:\! W A_i C_i \,|\,  R_L^1 Z R_W^1 R_W^2)_{\rho_i} \\
    & \leq \quad \rI (R_K R_J^1 R_S^1\!:\! ZWA_i C_i \,|\,  R_L^1 R_W^1 R_W^2)_{\rho_i} \\
    & = \quad \sum_{p \leq i,\ p~\mathrm{even}} \rI (R_K R_J^1 R_S^1
\!:\!
              C_p \,|\,  ZW A_p R_L^1 R_W^1 R_W^2)_{\rho_p} \\
    & \qquad \mbox{ }  - \sum_{p \leq i,\ p~\mathrm{odd}} \rI (R_K R_J^1
R_S^1 \!:\! C_p \,|\,  ZW A_p R_L^1 R_W^1 R_W^2)_{\rho_p} 
   + \rI(R_K R_J^1 R_S^1 \!:\! ZW \,|\,  R_L^1 R_W^1 R_W^2)_{\rho_0} \\
    & \leq \quad  \sum_{p \leq i,\ p~\mathrm{even}} \rI (R_K R_J^1 R_S^1
R_L^1 R_W^1 R_W^2 \!:\! C_p \,|\,  ZW A_p)_{\rho_p}
        + \rI(R_K R_J^1 R_S^1 \!:\! ZW R_W^1 R_W^2 \,|\,  R_L^1 )_{\rho_0} \\
    & = \quad \sum_{p \leq i,\ p~\mathrm{even}} \rI (R_K R_J^1 R_S^1
              R_L^1 R_W^1 R_W^2 \!:\! C_p \,|\,  ZW A_p)_{\rho_p} \\
    & \qquad \mbox{ }  + \rI(R_K R_J^1 R_S^1 \!:\! Z \,|\,  R_L^1 )_{\rho_0}
        + \rI(R_K R_J^1 R_S^1 \!:\! W R_W^1 R_W^2 \,|\,  Z R_L^1 )_{\rho_0} \\
    & = \quad \sum_{p \leq i,\ p~\mathrm{even}} \rI (R_K R_J^1 R_S^1
R_L^1 R_W^1 R_W^2 \!:\! C_p \,|\,  ZW A_p)_{\rho_p} \enspace.
\end{align*}
The first equality holds by the chain rule. The first inequality holds
because the first term evaluates to zero, and because the second term is
dominated by the subsequent expression (as may be seen by applying
the chain rule).
The second equality is from the information flow lemma, the second 
inequality follows from the chain rule and the non-negativity of the conditional mutual information, and the third 
equality is by the chain rule. The last equality follows because
$R_K R_J^1 R_S^1$ is independent of~$Z$ and because the registers $R_W^1
R_W^1 W$ of $\rho_0$ are in a pure state. 

The last term is seen to be upper bounded by $\QIC_{\rB \rightarrow \rA} (\Pi, \mu_0)$ by applying the data processing 
inequality to the $R$ register.
\end{proof}

The improved lower bound on $\QIC$ follows along the same lines as in
Section~\ref{sec:lbqic}, but we use Lemma~\ref{lem:impqicvsalice}
instead of Lemma~\ref{lem:qicvsalice} for even $i$'s in
Eq.~(\ref{eq:qicbavshi}). Then Eqs.~(\ref{eq:lbqicfirst}) to~(\ref{eq:lbqiclast}) become
\begin{align*}
\frac{1}{4} (1 - 2 \epsilon) \quad \leq \quad & \sum_{i~\textrm{odd}} \left( \frac{8 }{n}
\cdot \widetilde{\QIC}_{\rA \rightarrow \rB}^i (\Pi, \mu_0)
\right)^{1/2} + \sum_{i~\textrm{even}}
\left( 2
\cdot \widetilde{\QIC}_{\rB \rightarrow \rA}^i (\Pi, \mu_0) \right)^{1/2} \\
		\quad \leq \quad & \left( \frac{4 (M+1) }{n} 
\cdot \sum_{i~\textrm{odd}} \widetilde{\QIC}_{\rA \rightarrow \rB}^i (\Pi, \mu_0) \right)^{1/2} +
\left(  M \cdot \sum_{i~\textrm{even}} \widetilde{\QIC}_{\rB \rightarrow \rA}^i (\Pi, \mu_0) \right)^{1/2} \\
		\quad \leq \quad &  \left( \frac{2 (M+1)^2}{n} \cdot \QIC_{\rA \rightarrow
\rB} (\Pi, \mu_0) \right)^{1/2} +  \left(  \frac{M^2}{2} \cdot \QIC_{\rB
\rightarrow \rA} (\Pi, \mu_0) \right)^{1/2} \enspace,
\end{align*}
completing the proof of Theorem~\ref{th:formainlbaugind} for even~$M$.
The case of odd~$M$ is similar.

\bibliographystyle{alpha}
\bibliography{dyck-lb}

\appendix

\section{Relating Bob's states to $\QIC_{\rA \rightarrow \rB}$}
\label{app:jnbobqic}

Lemma~\ref{lem:qicvsbob} can be inferred from the proof of Lemma~4.9 in Ref.~\cite{JainN14}.
For completeness, we provide a proof using our notation.

\begin{lemma}
Given any even $n$, let $J$ and $L$ be random variables uniformly
distributed in $[n/2]$ and $[n] \setminus [n/2]$, respectively. Conditional on some 
value $\ell$ for $L$, let $Z$ be a random variable chosen uniformly at random in $\{0, 1 \}^\ell$. Then the following holds for any $M$-message safe quantum protocol $\Pi$ for the
Augmented Index function~$f_n \,$, for any odd $i \leq M$:
\begin{align*}
\frac{1}{n} \; \widetilde{\QIC}_{\rA \rightarrow \rB}^i (\Pi, \mu_0) \quad
\geq \quad
\frac{1}{16 } \; \mathbb{E}_{j \ell z\sim JLZ} \left[ \fh^2 \!\left(\rho_{i, \ell
z j z^\prime}^{  B_i C_i}~,~\rho_{i, \ell z^{(\ell)} j z^\prime}^{  B_i
C_i} \right)\right] \enspace,
\end{align*}
with $\rho_{i, \ell z j z^\prime}$ defined by Eqs.~(\ref{eq:rhoilzk}) and
(\ref{eq:purerhoilzk}).
\end{lemma}
\begin{proof}
We start with the following chain of inequalities.
\begin{align}
\nonumber
\MoveEqLeft  \widetilde{\QIC}_{\rA \rightarrow \rB}^i (\Pi, \mu_0) \\
\nonumber
    & =  \quad \rI (X \!:\! B_i C_i \,|\,  K X [1, K]) \\
\nonumber
    & \geq \quad \frac{1}{2} \; \rI (X \!:\! B_i C_i \,|\,  J X [1, J]) \\
\nonumber
    & = \quad \frac{1}{2} \cdot \frac{2}{n} \sum_{j \leq n/2} \frac{1}{2^j}
          \sum_{z^\prime} \rI (X[j+1, n]\!:\! B_i C_i \,|\,
          J = j, X[1, j] = z^\prime ) \\
\nonumber
    & \geq \quad \frac{1}{n} \sum_{j \leq n/2} \frac{1}{2^j}
             \sum_{z^\prime} \rI (X[n/2+1, n]\!:\! B_i C_i \,|\, 
             X[j+1, n/2], J = j, X[1, j] = z^\prime ) \\
\nonumber
    & = \quad \frac{1}{n} \sum_{j \leq n/2, \ell > n/2} \frac{1}{2^j}
          \sum_{z^\prime} \rI (X_\ell\!:\! B_i C_i \,|\, 
          X[j+1, \ell-1], J = j, X[1, j] = z^\prime ) \\
\label{eq:jnfirst}
    & = \quad \frac{1}{n} \sum_{j \leq n/2, \ell > n/2}
          \frac{1}{2^{\ell-1}} \sum_{z[1, \ell-1]}
          \rI (X_\ell\!:\! B_i C_i \,|\,   J = j, X[1, \ell-1] = z[1, \ell-1] )
          \enspace.
\end{align}
The first equality holds by definition, the first inequality holds
because we can generate the classical random variable~$K$ by setting it
equal to $J$ with probability one half (and then equal to $L$ also with
probability one half), the second equality follows by expanding over $J$
and $X[1, J]$, the second inequality follows by the chain rule and
non-negativity of mutual information, the third equality holds by the
chain rule, and the last equality follows by expanding over $X[j+1, \ell-1]$.

We also get the following bound by the Average Encoding Theorem
(Lemma~\ref{lem:avenc}) and the
weak triangle inequality for $\fh^2$ (cf.\ Eq.~\eqref{eq-weak-triangle}), with $\rho_{i, \ell z[1, \ell-1] j
z^\prime}$ being the state $\rho_{i, \ell z[1, \ell-1] x_\ell  j z^\prime}$ in register $B_i C_i$ averaged over $X_\ell \,$.
\begin{align*}
\MoveEqLeft  \rI (X_\ell \!:\! B_i C_i \,|\,  J = j, X[1, \ell-1] = z[1,
\ell-1]) \\
    & \geq \quad \frac{1}{2} \; \fh^2 \!\left(\rho_{i,  \ell z j z^\prime }^{B_i C_i}~,
             ~\rho_{i, \ell z[1, \ell-1] j z^\prime }^{B_i C_i} \right) 
             +  \frac{1}{2} \; \fh^2 \!\left(
             \rho_{i, \ell z^{(\ell)} j z^\prime }^{B_i C_i}~,
             ~\rho_{i, \ell z[1, \ell-1] j z^\prime }^{B_i C_i} \right) \\
    & \geq \quad \frac{1}{4} \; \fh^2 \!\left(\rho_{i,  \ell z j z^\prime }^{B_i C_i}~,
             ~\rho_{i, \ell z^{(\ell)} j z^\prime}^{B_i C_i} \right) \enspace.
\end{align*}
Taking expectation over $JLZ$ in the above inequality and expanding, we
get the desired result by comparing it with Eq.~(\ref{eq:jnfirst}).
\begin{align*} 
\MoveEqLeft { \mathbb{E}_{j \ell z\sim JLZ} \left[ \fh^2 \!\left(\rho_{i, \ell z j
z^\prime}^{  B_i C_i}~,~\rho_{i, \ell z^{(\ell)} j z^\prime}^{  B_i C_i}
\right)\right] } \\
     & \leq  \quad  4 \; \mathbb{E}_{j \ell z\sim JLZ} \;
             \rI (X_\ell \!:\! B_i C_i \,|\,  J = j, X[1, \ell-1] = z[1, \ell-1]) \\
     & = \quad  4 \left( \frac{2}{n} \right)^2 \sum_{j \leq n/2, \ell > n/2}
            \frac{1}{2^{\ell-1}} \sum_{z[1, \ell-1]}
            \rI (X_\ell \!:\! B_i C_i \,|\,  J = j, X[1, \ell-1] = z[1, \ell-1]) \\
     & \leq \quad \frac{16}{n} \; 
           \widetilde{\QIC}_{\rA \rightarrow \rB}^i (\Pi, \mu_0) \enspace.
\end{align*}
\end{proof}

\section{Information Flow Lemma}
\label{app-ifl}

We use the following bound on the transfer of information in interactive quantum protocols, 
obtained in Ref.~\cite{LT17-information-flow,LT17-information-flow-arxiv}. We include a proof for completeness.

\begin{lemma}
\label{lem:rdindeplb}
	Given a protocol $\Pi$, an input state $\rho$ with purifying register $R$ with arbitrary 
decompositions $R = R_a^\rA R_b^\rA R_c^\rA = R_a^\rB R_b^\rB R_c^\rB$,
the following identities hold:
\begin{align*}
\MoveEqLeft
\sum_{i \geq 0} \rI (R_a^\rB \!:\! C_{2i + 1} \,|\,  R_ b^\rB B_{2i
+ 1} ) - \sum_{i \geq 1} \rI (R_a^\rB \!:\! C_{2i} \,|\,  R_b^\rB B_{2i}) \\
    & = \quad \rI (R_a^\rB \!:\! B_{\rout} B^\prime \,|\,  R_b^\rB) -
\rI (R_a^\rB \!:\!
        B_{\rin} \,|\,  R_b^\rB) \enspace, \\
    & \\
\MoveEqLeft \sum_{i \geq 0} \rI (R_a^\rA \!:\! C_{2i + 2} \,|\,  R_ b^\rA A_{2i
+ 2} ) - \sum_{i \geq 0} \rI (R_a^\rA \!:\! C_{2i + 1} \,|\,  R_b^\rA A_{2i
+ 1}) \\
    & = \quad \rI (R_a^\rA \!:\! A_{\rout} A^\prime \,|\,  R_b^\rA) -
\rI (R_a^\rA \!:\!
        A_{\rin} \,|\,  R_b^\rA) \enspace.
\end{align*}
\end{lemma}

\begin{proof}
We focus on the first identity, that for the messages received by Bob; the
identity for the messages received by Alice follows similarly. In the
rest of the proof, we omit the superscripts on the purifying registers;
they are meant to be~$\rB$. 

We show that
\begin{align*}
\MoveEqLeft \rI (R_a \!:\! B_{2k + 1} C_{2k +1} \,|\,  R_b) \\
    & = \quad
\sum_{0 \leq i \leq k} \rI (R_a \!:\! C_{2i + 1} \,|\,  R_ b B_{2i + 1} ) -
\sum_{1 \leq i \leq k} \rI (R_a \!:\! C_{2i} \,|\,  R_b B_{2i}) + \rI
(R_a \!:\!
B_{\rin} \,|\,  R_b)
\end{align*}
by induction on $k$, with~$2k+1 \leq M$.
If $M$ is odd, Bob receives the last message 
and  $\rI (R_a \!:\! B_{\rout} B^\prime \,|\,  R_b) 
= \rI (R_A \!:\! B_{M} C_{M} \,|\,  R_b)$, and the result follows.
If $M$ is even 
and Bob sends the last message, the result follows since $B_M = B_{\rout} B^\prime$ 
and $\rI (R_a \!:\! B_M \,|\,  R_b) = \rI (R_a \!:\! B_{M - 1} C_{M - 1}
\,|\,  R_b)  - \rI (R_a \!:\! C_M \,|\,  R_b B_M) $ 
by using the chain rule and isometric invariance under the map that
takes~$B_{M-1} C_{M-1} \rightarrow B_M C_M$. 

The base case for the induction follows from
\begin{align*}
	\rI (R_a \!:\! B_1 C_1 \,|\,  R_b) 	& \quad = \quad \rI (R_a
\!:\! B_1 \,|\,  R_b)  + \rI(R_a \!:\! C_1 \,|\,   R_ b B_1) \\
				& \quad = \quad \rI (R_a \!:\! B_{\rin} \,|\,
R_b)  + \rI(R_a \!:\! C_1 \,|\,   R_ b B_1) \enspace.
\end{align*}
Here, the first equality holds by the chain rule. The second holds
because $B_1 = B_0 = B_{\rin} T_B$, and
because the state in~$T_B$ is in tensor product with the initial state in the
registers~$R_a R_b B_{\rin}$.

For the induction step, we have 
\begin{align*}
\MoveEqLeft
\rI(R_a \!:\! B_{2k+3} C_{2k + 3} \,|\,  R_b ) \\
    = \quad & \rI (R_a \!:\! B_{2k+3} \,|\,  R_b)
        + \rI(R_a \!:\! C_{2k+2} \,|\,  R_b B_{2k+2}) 
        + \rI (R_a \!:\!  C_{2k+3} \,|\,  R_b B_{2k+3} ) \\
            & \mbox{} - \rI(R_a \!:\! C_{2k+2} \,|\,  R_b B_{2k+2}) \\
    = \quad & \rI(R_a: B_{2k +2} C_{2k+2} \,|\,  R_b )
        + \rI (R_a \!:\!  C_{2k+3} \,|\,  R_b B_{2k+3} ) 
        - \rI(R_a \!:\! C_{2k+2} \,|\,  R_b B_{2k+2}) \\
    = \quad & \rI(R_a: B_{2k +1} C_{2k+1} \,|\,  R_b )
        + \rI (R_a \!:\!  C_{2k+3} \,|\,  R_b B_{2k+3} ) 
        - \rI(R_a \!:\! C_{2k+2} \,|\, R_b B_{2k+2}) \enspace,
\end{align*}
in which the first equality holds by the chain rule and by adding and subtracting the same term,
the second also holds by the chain rule and because $B_{2k+3} = B_{2k+2}$, and
the third holds by the isometric invariance under the map that takes~$B_{2k+1} C_{2k+1} \rightarrow B_{2k+2} C_{2k+2}$.
The induction step follows by comparing terms.
\end{proof}

\end{document}